\documentclass[11pt, letterpaper]{article}

\usepackage{CJKutf8}

\usepackage[whole]{bxcjkjatype}

\usepackage{algorithm}
\usepackage[noend]{algorithmic}

\usepackage{cite} 
\usepackage[pdftex]{graphicx}
\usepackage{threeparttable}
\usepackage[margin=1in]{geometry}
\usepackage{amsmath,amssymb,amsfonts,setspace} 
\usepackage{amsthm} 
\usepackage{latexsym}
\usepackage{comment}
\usepackage{bm}
\usepackage{enumerate}
\usepackage{bookmark} 
\usepackage{type1cm} 
\usepackage{color} 
\usepackage{algorithm}
\usepackage[noend]{algorithmic}

\theoremstyle{plain}
\newtheorem{theorem}{Theorem}
\newtheorem{lemma}{Lemma}
\newtheorem{corollary}{Corollary} 
\newtheorem{definition}{Definition}

\newcounter{cntLemmaNumber}

\newcounter{cntTheoremNumber}

\newcommand{\MV}{\mathsf{MV}}
\newcommand{\PART}{\ensuremath{\mathsf{PART}}}
\newcommand{\FUNC}{\ensuremath{\mathsf{EXTPATH}}}

\newcommand{\SAPath}{\mathsf{A}}

\newcommand{\dist}{\mathrm{dist}}

\newcommand{\Odd}{\mathrm{odd}}
\newcommand{\Even}{\mathrm{even}}
\newcommand{\OL}[1]{\overline{#1}}
\newcommand{\Parity}{\{\Odd, \Even\}}

\newcommand{\level}{\ensuremath{\mathsf{level}}}
\newcommand{\parent}{\ensuremath{\mathsf{p}}}

\newcommand{\lca}{\ensuremath{\mathsf{lca}}}
\newcommand{\out}{\ensuremath{\mathsf{out}}}

\newcommand{\sumlevel}{\ensuremath{\mathsf{level_{sum}}}}
\newcommand{\maxlevel}{\ensuremath{\mathsf{level_{max}}}}

\newcommand{\Mmax}{\mu(G)}
\newcommand{\AMmax}{\hat{\mu}}

\allowdisplaybreaks[2]

\title{A Nearly Linear-Time Distributed Algorithm \\ for Maximum Cardinality Matching}
\author{Taisuke Izumi\thanks{Osaka University. Email: izumi.taisuke.ist@osaka-u.ac.jp, \{n-kitamura, yutaro.yamaguchi\}@ist.osaka-u.ac.jp}
\and Naoki Kitamura$^*$
\and Yutaro Yamaguchi$^*$
}

\date{}

\begin{document}

\maketitle


\begin{abstract}
In this paper, we propose a randomized $\tilde{O}(\Mmax)$-round algorithm for the maximum cardinality matching problem in the CONGEST model, where $\Mmax$ means the maximum size of a matching of the input graph $G$. 
The proposed algorithm substantially improves the current best worst-case running time.
The key technical ingredient is a new randomized algorithm of finding an augmenting path of length $\ell$ with high probability within $\tilde{O}(\ell)$ rounds, which positively settles an open problem left in the prior work by 
Ahmadi and Kuhn [DISC'20].

The idea of our augmenting path algorithm is based on a recent result by Kitamura and Izumi [IEICE Trans.'22], which efficiently identifies a sparse substructure of the input graph containing an augmenting path, following a new concept called \emph{alternating base trees}.
Their algorithm, however, resorts in part to a centralized approach of collecting the entire information of the substructure into a single vertex for constructing a long augmenting path.
The technical highlight of this paper is to provide a fully-decentralized counterpart of such a centralized method. To develop the algorithm, we prove several new structural properties of alternating base trees, which are of independent interest.
\end{abstract}

\section{Introduction}

\subsection{Background and Our Result}
In this paper, we consider the maximum matching problem\footnote{Throughout this paper, we only consider the unweighted problem, and thus simply refer to a maximum cardinality matching as a \emph{maximum} matching.} for general graphs in the CONGEST model, a standard computational model for distributed graph algorithms.
The network is modeled as an undirected graph 
$G = (V, E)$ of $n$ vertices (nodes) and 
$m$ edges (communication links). Each vertex executes a deployed algorithm following round-based synchrony, and each edge 
can transfer a message of $O(\log n)$ bits per round bidirectionally. In distributed settings, any graph problem admits 
the trivial solution of aggregating all the topological information of $G$ into a single vertex (and solving the problem 
by any centralized algorithm). Due to bandwidth limitation of the CONGEST model, such a trivial solution requires $\Theta(m)$ rounds, which may be $\Theta(n^2)$. Hence, the technical challenge in the CONGEST model is to design an algorithm such that each node outputs the solution without whole input information at hand.

The maximum matching problem for general graphs is often referred to as the ``hardest'' problem of those admitting 
a polynomial-time solution in the context of centralized graph algorithms. A similar situation can be seen in the CONGEST model: Until recently, it was unclear if one could construct an algorithm faster than the trivial $O(n^2)$-round algorithm or not. Looking at the lower-bound side, it has been proved that any CONGEST algorithm of computing an exact solution of  the maximum matching problem requires $\Omega(\sqrt{n} + D)$ rounds in the worst case even if the input is restricted to bipartite 
graphs~\cite{AKO18}, where $D$ is the diameter of the input graph. One may think that this bound can be improved when considering 
general graphs. However, it is also shown that such a lower 
bound cannot be improved even for general graphs as long as the proof is based on the two-party 
communication complexity argument, which is a dominant technique to prove lower bounds in the CONGEST model\cite{BCDELP19}. 
In 2022, Kitamura and Izumi~\cite{KI22} proposed an algorithm running within $\tilde{O}(\Mmax^{1.5})$ rounds, where $\Mmax$ represents 
the size of the maximum matching of the input graph $G$. It is the first algorithm beating the worst-case bound of $O(n^2)$ 
rounds. However, the current gap between the best upper and lower bounds is still large.

The primary contribution of this paper is the development of a new CONGEST algorithm for the maximum matching problem running in nearly linear time. The main theorem is stated below:

\begin{theorem}
\label{thm:main}
There exists a randomized CONGEST algorithm to compute a maximum  matching of a given graph with high 
probability\footnote{Throughout this paper, the phrase ``with high probability'' means it
succeeds with probability at least $1 - 1/n^{c}$ for an arbitrary constant 
$c > 1$.} that terminates within $\tilde{O}(\Mmax)$ rounds.
\end{theorem}

\subsection{Technical Highlight}
Our algorithm is built on the top of several known techniques. We first present an outline of
those techniques, and then explain the key technical idea newly proposed in this paper.

\paragraph{Edmonds' Blossom Algorithm} 

The starting point is the well-known centralized algorithm by Edmonds, the
so-called \emph{blossom algorithm}~\cite{Edmonds}. This algorithm iteratively finds \emph{augmenting paths} and improves 
the current matching size using the identified augmenting paths. If an augmenting path is found, 
the current matching is augmented by flipping the labels of matching and non-matching edges along the path. 
The seminal theorem by Berge~\cite{Berge57} states that the current matching is the maximum if and only 
if there exists no augmenting path. Hence, 
the central difficulty of the maximum matching problem lies in the efficient detection of augmenting paths.

Whereas a simple approach of expanding alternating paths in the BFS-like manner correctly works 
in bipartite graphs, it fails in general graphs because of the existence of odd cycles, 
which destroys the very useful property holding in bipartite graphs that any subpath 
of a shortest alternating path is also a shortest alternating path. Following \cite{Vazirani12,vazirani2020proof}, 
we refer to this property as the \emph{BFS-honesty}. Edmonds' algorithm basically 
works in the same way as the case of bipartite graphs, i.e., starting from unmatched vertices, 
it makes a BFS-like forest grow up by choosing matching and non-matching edges alternately. If 
the forest growth process encounters a \emph{blossom}, an odd cycle and an alternating path 
from the tree source to
the vertex incident to two non-matching edges in the cycle, the algorithm contracts its 
odd-cycle part into a single vertex; subsequently, the forest growth process continues 
in the resultant graph. An important observation is that the existence of augmenting paths is 
preserved by this contraction operation. Once an augmenting path of the contracted graph is found, an augmenting 
path of the original graph must be recovered from it. Roughly, when the augmenting path 
finally found contains a vertex associated with a contracted cycle, that vertex
is expanded to an appropriate subpath in the cycle. This process is not so simple 
because contracted cycles can be nested (i.e., some odd cycle may contain a vertex associated with  
a contracted cycle), but certainly possible.

Since Edmonds' algorithm is essentially sequential, it is highly challenging 
to obtain its efficient distributed implementation in the CONGEST model. There are three major 
obstacles in the implementation. 1) The iterative improvement 
of matchings: It can be iterated $\Theta(n)$ times. 2) The construction of BFS-like forests:
The depth of a tree in the forest may become $\Theta(n)$. 3) Graph contraction: The simulation of 
the contracted graph on the original message-passing topology can incur expensive time 
cost, and the path recovery process is inherently sequential, particularly when contracted 
cycles are nested. As a result,
even it is nontrivial to compute an augmenting path with a running time linearly dependent 
on its length. In fact, a natural distributed implementation of Edmonds' approach incurs $\Theta(\ell^2)$ rounds 
for finding an augmenting path of length $\ell$. Since the length of a shortest augmenting path can become 
$\Theta(n)$, such an implementation is not sufficient to beat the trivial bound of $O(n^2)$ rounds. 

\paragraph{Framework Based on Hopcroft--Karp Analysis}
In addition to faster augmenting-path detection, efficient distributed implementation 
of Edmonds' algorithm includes yet another challenge that the iterative improvement 
can be repeated $\Theta(n)$ times. To design a nearly linear-time algorithm, we 
must attain the amortized $\tilde{O}(1)$-round complexity of finding one augmenting path,
which looks a big hurdle. Fortunately, it is achieved by employing 
the classic Hopcroft--Karp structural property on the length of shortest augmenting paths:
\begin{theorem}[Hopcroft and Karp~\cite{HK73}]
\label{thm:hk}
Let $G = (V, E)$ be an undirected graph, and $\Mmax$ be the maximum matching size of $G$. Given a matching $M \subseteq E$ of $G$ with $|M| \leq \Mmax - k$, there always exists an augmenting path of length less 
than $\lfloor2\Mmax/k\rfloor$. 
\end{theorem}
Following this theorem, it suffices to find an augmenting path of length $\ell$ within 
$\tilde{O}(\ell n^{1-\epsilon})$ rounds for improving the trivial $O(n^2)$-round upper bound. By an elementary analysis, 
such an algorithm works as an amortized $\tilde{O}(n^{1 - \epsilon})$-round algorithm. In particular, the algorithm with 
$\epsilon = 1$ (i.e. $\tilde{O}(\ell)$-round algorithm) is the best-possible goal along this approach, which 
attains the amortized optimal $\tilde{O}(1)$-round complexity.

\paragraph{Distributed Maximum Matching Verification by Ahmadi and Kuhn}
Motivated by the aforementioned framework, Ahmadi and Kuhn~\cite{AK20} provided 
an interesting idea toward $\tilde{O}(\ell)$-round augmenting path algorithms. They 
actually did not present any construction algorithm, but a \emph{verification} algorithm 
of determining if a given matching is maximum or not. To be more precise, the proposed algorithm 
provides each vertex with information on the lengths of shortest odd- and even-alternating paths from a closest 
unmatched vertex, where an odd-alternating (resp.\ even-alternating) path is an alternating path with 
an odd (resp.\ even) length. The running time of the algorithm is parameterized by $\ell$, and all the lengths 
till $\ell$ are reported within $O(\ell)$ rounds.

They also shows that the computed length information can be used to identify a subgraph
which contains exactly two unmatched vertices admitting a shortest augmenting path connecting them. It implies that
the quest for $\tilde{O}(\ell)$-round augmenting path construction algorithms can be considered in a much more moderate 
setting, where the two endpoints of the target shortest augmenting path has been known, and every vertex knows the lengths of 
shortest odd- and even-alternating paths from the endpoints. However, it should be emphasized that constructing 
an augmenting path is still challenging even in this setting. Due to a lack of the BFS-honesty, 
subpaths of a shortest-alternating path are not necessarily shortest alternating paths. Then, 
the straightforward path-tracing approach based on distance information does not work.

\paragraph{Semi-Distributed Augmenting Path Construction by Kitamura and Izumi}

Kitamura and Izumi~\cite{KI22} addressed the limitation of the aforementioned research line. They proposed
an algorithm which finds an augmenting path within $\tilde{O}(\ell \sqrt{n})$ rounds on the top of the Ahmadi--Kuhn framework,
if the current matching admits an augmenting path of length $\ell$.
Their approach is to construct a \emph{sparse certificate}, which is a sparse (i.e., containing $O(n)$ edges\footnote{The original paper~\cite{KI22} actually constructs the certificate of $O(\Mmax)$ edges, and attains the algorithm running within 
$\tilde{O}(\ell \sqrt{\Mmax})$ rounds. Here we explain 
the idea with the simplification of $\Mmax = \Theta(n)$.}) subgraph of $G$ preserving the reachability from 
unmatched vertices by alternating paths. Specifically, a sparse certificate contains a $\theta$-alternating path from an unmatched vertex 
$s$ to some vertex $t$ for $\theta \in \{\Odd, \Even\}$ if and only if the original graph admits such a path. A node can collect all the information on the sparse certificate within $O(n)$ rounds, trivially allowing the centralized solution of finding augmenting paths. The $\tilde{O}(\ell \sqrt{n})$ running time 
is attained by using this algorithm if $\ell \geq \sqrt{n}$, or the naive $O(\ell^2)$-round algorithm otherwise.

It is intuitively understandable that such a certificate always exists: 
Looking at Edmonds' algorithm, the augmenting path finally found 
in the contracted graph is expanded to a path in the original graph whose edges lie in the constructed BFS-like 
forest and all the contracted odd cycles. In other words, the union of the BFS-like forest and all the contracted cycles, 
whose number is at most $n$, becomes a sparse certificate. 
The key ingredient of the Kitamura--Izumi algorithm is an efficient identification of such a sparse certificate in 
the setting of the Ahmadi--Kuhn framework. It first constructs a rooted spanning tree 
$T$ called an \emph{alternating base tree}, which corresponds to the BFS-like forest by Edmonds, within $O(1)$ rounds 
using the distance information provided by the Ahmadi--Kuhn framework (it should be noted that $T$ is a single tree because every
vertex already knows the endpoints of the target augmenting path with the aid of the Ahmadi--Kuhn framework, and thus it suffices to 
start the tree growth process from one of the endpoints).
To transform the alternating base tree into a sparse certificate, they introduced the notion of \emph{edge levels} and showed that 
the desired structure is obtained by adding one minimum-level outgoing non-tree edge 
for each subtree of $T$. It is also shown that the computation of edge levels 
and the identification of minimum-level outgoing non-tree edges are efficiently implemented.

\paragraph{Our Technique}
While our algorithm is based on the techniques by Kitamura and Izumi,  
it includes many new ideas and insights regarding the problem. A big leap from Kitamura--Izumi is that we provide a 
\emph{fully-decentralized} path-construction algorithm, which means that no vertex needs to collect whole information of the target augmenting path at hand. Obviously, it needs to resolve the following question: how is the task of finding long augmenting paths decomposed 
into that of finding shorter subpaths? No mechanism of addressing this question has been known so far in the CONGEST model. In addition, the (advanced) use of the 
Kitamura--Izumi algorithm also faces another obstacle: the constructed sparse certificate certainly preserves the existence of the target augmenting path, but does not necessarily preserve its length.
Hence, it is also unlikely that one can benefit from their sparsification technique, except for 
the centralized collection approach.

The primary technical contribution of this 
paper is to demonstrate that both obstacles can be removed just by a small modification of their sparsification algorithm.
More precisely, we use a slightly modified definition of edge levels, and consider the same 
sparse certificate construction. Under our new definition, it can be shown that the constructed certificate 
preserves at least one \emph{shortest} $\theta$-alternating path for any relevant pair of vertices and any 
$\theta \in \{\Odd, \Even\}$ (if it exists), including a shortest augmenting path. 
Furthermore, the preserved alternating paths are equipped with an
important property: let $T$ be the alternating base tree, $s$ be a vertex, and $t$ be a descendant of $s$ in $T$ 
and reachable from $s$ by a $\theta$-alternating path. Then it is guaranteed that 
at least one preserved shortest $\theta$-alternating path $P$ from $s$ to $t$ is split into the prefix $P_1$ and the suffix $P_2$ 
respectively lying at the outside and the inside of the subtree of $T$ rooted by $t$. Intuitively, this property 
implies that the task of finding a shortest $\theta$-alternating path $P$ from $s$ to $t$ is decomposed 
into two independent subtasks of finding $P_1$ and $P_2$ at the outside/inside of the subtree.
Following this observation, we develop a
surprisingly simple recursive algorithm of constructing a shortest augmenting path, which runs in time linearly 
dependent on the output length.

To show the correctness of our algorithm, many rich structural properties of alternating base trees are newly proved in this paper.
We believe that these insights are of independent interest, and provide a very useful toolbox toward further improvement of algorithms, which will run in sublinear time, and other related settings such as approximate and/or weighted variants.

\subsection{Related Works}
There have been many known results for the maximum matching problem in divergent 
settings~\cite{II86,ABI86,FTR06,KMW16,BCGS17,Harris19,CS22,HS23}. Table~\ref{tab:MM} summarizes the known (approximate) maximum matching algorithms 
in distributed settings, with focus on relevant to our result. 
Many literatures have considered approximate maximum matchings in both CONGEST and 
the more relaxed LOCAL model, wherein no restriction of bandwidth is incurred. The main 
interest of those studies is the localization of the computation process. Typically, 
the results along this line exhibits the trade-off between the approximation factor and 
the running time, and some algorithms attain the $(1- \epsilon)$-approximation factor, which can be seen as an exact algorithm by taking 
$\epsilon < 1/n$. However, most of those results in the CONGEST model
have the running time exponentially depending on $1/\epsilon$, and thus their performance as an exact
algorithm is very poor. One exception is the recent breakthrough result by Fischer, Slobodan, and Uitto~\cite{FMJ22}, which 
provides a $\mathrm{poly}(1/\epsilon, \log n)$-time $(1 - \epsilon)$-approximate algorithm. 
Recently, a similar result for weighted matchings was proposed by Huang and Su~\cite{HS23}. 
However, 
the hidden exponents of $1/\epsilon$ for those algorithms are quite larger than two, and thus 
it is not sufficient to break the trivial $O(n^2)$ bound. The paper by Huang and Su~\cite{HS23} 
also provides an extensive survey on distributed approximate maximum matching algorithms.

Focusing on the exact computation, there are three algorithms of attaining nontrivial upper bounds.
Ben-Basat, Kawarabayashi, and Schwartzman~\cite{BKS18} presents a CONGEST algorithm whose running time is parameterized
by $\Mmax$. The running time of the algorithm is $\tilde{O}(\Mmax^2)$. They also proved
the lower bound of $\Omega(\Mmax)$ in the LOCAL model. A generalized $\Omega(1 / \epsilon)$ lower bound for the $(1 - \epsilon)$-approximate maximum matching is also presented. While these bounds apply
to the CONGEST model, the worst-case instance of providing these bounds has the diameter $\Theta(n)$. 
Parameterizing this lower bound with diameter $D$, it only exhibits the bound of $\Omega(D)$.
Since the $\tilde{\Omega}(\sqrt{n} + D)$-round bound we mentioned 
applies to the hard-core instances of $D = O(\log n)$, their bounds are incomparable. 
An improved
$\tilde{O}(n)$-round upper bound for bipartite graphs was proposed by Ahmadi, Kuhn, and Oshman~\cite{AKO18}.
This is also the first paper utilizing the framework based on the analysis of Hopcroft--Karp.
It is worth remarking that no $o(n)$-round algorithm is known even for bipartite matchings. 
The third algorithm is the one by Kitamura and Izumi~\cite{KI22} attaining $\tilde{O}(\Mmax^{1.5})$ rounds. 
Until our result, this was only the known result faster than the trivial $O(n^2)$-round algorithm. 

We also remark centralized exact maximum matching algorithms. Edmonds' blossom algorithm is
the first centralized polynomial-time solution for the maximum matching problem~\cite{Edmonds,Edmonds2}.
Hopcroft and Karp~\cite{HK73} proposed a phase-based algorithm of finding multiple augmenting paths. Their algorithm finds a maximal set of pairwise disjoint shortest augmenting paths in each phase. They showed 
that $O(\sqrt{n})$ phases suffice to obtain a maximum matching and proposed an $O(m)$-time per phase implementation for bipartite graphs. Several studies have reported phase-based algorithms for general graphs that attain $O(\sqrt{n}m)$ time~\cite{Blum1990,MV1980,Vazirani12,vazirani2020proof,GT91}. While our algorithm 
actually finds multiple disjoint augmenting paths in parallel, it cannot benefit from such a parallelism 
because finding the last augmenting path can take $\Omega(n)$ rounds after all. It is an interesting open question 
if the technique of the phase-based algorithms above can be utilized for obtaining much faster CONGEST algorithms or not.

\begin{table*}[tb]
\caption{Lower and upper bounds for (semantically) exact maximum matching algorithms in the CONGEST model}
\center
\label{tab:MM}
\renewcommand{\arraystretch}{1.25}
\begin{tabular}{ l c c c}
 & Time Complexity &Approximation Level  & Remark\\
\hline
Ben-Basat et al.~\cite{BKS18} & $\Omega(\Mmax)$ & exact &LOCAL\\
Ben-Basat et~al.~\cite{BKS18}& $\Omega\left(\frac{1}{\epsilon}\right)$ & $1-\epsilon$ &LOCAL\\
Ben-Basat et~al.~\cite{BKS18} & $\tilde{O}(\Mmax^2)$ &exact & \\
Kitamura and Izumi~\cite{KI22} & $\tilde{O}(\Mmax^{1.5})$ &exact & \\
Ahmadi et~al.~\cite{AKO18}& $\tilde{O}(\Mmax)$& exact & bipartite\\
Fischer et~al.~\cite{FMJ22}& $O(\mathrm{poly}(\epsilon, \log n))$ & $1-\epsilon$ & \\
Huang and Su~\cite{HS23}& $O(\mathrm{poly}(\epsilon, \log n))$ & $1-\epsilon$ & weighted \\

Ahmadi et~al.~\cite{AKO18}& $O\left(\frac{\log^2 \Delta+\log^{*}n}{\epsilon^2}\right)$& $1-\epsilon$ & bipartite\\
Ahmadi et~al.~\cite{AKO18}& $\tilde{\Omega}(\sqrt{n} + D)$ & exact & bipartite\\

\textbf{Our result}& \boldmath{$\tilde{O}(\Mmax)$} & \textbf{exact} &\\
\hline
\end{tabular}
\end{table*}

\section{Preliminaries}
\label{sec:pre}

\subsection{Notations and Terminologies}

We denote the vertex set and edge set of a given graph $G$ by $V(G)$ and
$E(G)$, respectively. The diameter of a graph $G$ is denoted by $D(G)$. 
We denote by $I_G(v)$ the set of edges incident to $v$ in $G$.
Throughout this paper, we only consider simple undirected graphs, and an edge connecting two vertices $u$ and $v$ is denoted by $\{u, v\}$.

A \emph{path} of $G$ is an alternating sequence $P = v_0, e_1, v_1, e_2, \dots, e_\ell, v_\ell$ of vertices and edges such that all the vertices are distinct and $e_i = \{v_{i-1}, v_i\} \in E(G)$ for every $1 \leq i \leq \ell$.
A path is often regarded as a subgraph of $G$.
For a path $P = v_0, e_1, v_1, \dots, v_\ell$ of $G$ and an edge $e = \{v_\ell, u\} \in E(G)$ satisfying $u \not\in V(P)$, we denote by $P \circ e$ or $P \circ u$ the path obtained by adding $e$ and $u$ to the tail of $P$.
Similarly, for another path $P'$ starting at $v_\ell$ with $V(P) \cap V(P') = \{v_\ell\}$, we denote by $P \circ P'$ the path obtained by concatenating $P'$ to the tail of $P$ (without duplication of $v_\ell$).
The inversion of the path 
$P = v_0, e_1, v_1, \dots, v_\ell$ (i.e., the path $v_\ell, e_\ell, v_{\ell - 1}, e_{\ell -1}, \dots, v_0$) is denoted 
by $\overline{P}$. The length $\ell$ of a path $P$ is denoted by $|P|$. Given a path $P$ containing two vertices $u$ and $v$ 
such that $u$ precedes $v$, we denote by $P[u, v]$ the subpath of $P$ starting from $u$ and terminating with $v$.

For a graph $G$, a \emph{matching} $M\subseteq E(G)$ is a set of edges that do not share endpoints.
We denote by $\Mmax$ the maximum size of a matching of a graph $G$. 
In the following definitions, we fix a matching $M$ of $G$. A vertex $v$ is called \emph{matched} 
if $M$ intersects $I_G(v)$, and \emph{unmatched} otherwise. A path $P = v_0, e_1, v_1, e_2, \dots, e_\ell, v_\ell$ is called 
\emph{alternating} if exactly one of $e_{i-1}$ and $e_i$ is in $M$ for every $1 < i \le \ell$.
An \emph{augmenting path} is an alternating path connecting two different unmatched vertices. 
We say that $(G, M)$ has an augmenting path if there exists an augmenting path in $G$ 
with respect to $M$.
For notations $V(G)$, $E(G)$, $D(G)$, and $\Mmax$, we often use the shorthand without argument $G$ unless 
it causes any ambiguity.

Given a subgraph $H$ of $G$ and a matching $M$ of $H$, an alternating path with an 
odd (resp.\ even) length is called an \emph{odd-alternating path} (resp.\ \emph{even-alternating 
path}). For $u, v \in V(H)$, $v$ is called \emph{odd-reachable} (resp.\ \emph{even-reachable}) 
from $u$ in $H$ if there exists an odd-alternating (resp.\ even-alternating) path from $u$ to $v$. A vertex $v$ is simply called \emph{reachable} from $u$ if it is odd-reachable or even-reachable from $u$.
We denote the length of the shortest odd (resp.\ even)-alternating path from $u$ to $v$ in $H$ by $\dist^{\Odd}_H(u, v)$ (resp.\ $\dist^{\Even}_H(u,v)$). If $v$ is not odd (resp.\ even)-reachable from $u$, 
we define $\dist^{\Odd}_H(u,v) = \infty$ (resp.\ $\dist^{\Even}_H(u,v) = \infty$).
For $\theta \in \Parity$, $\OL{\theta}$ represents the parity different from $\theta$. The parity 
$\theta \in \Parity$ satisfying $\dist^{\theta}_H(u, v) < \dist^{\OL{\theta}}_H(u, v)$ (i.e., the parity of the 
shortest alternating path from $u$ to $v$) is denoted by $\gamma_H(u, v)$. Notice that it is well-defined for any reachable 
$v$ because $\dist^{\theta}_H(u, v) \neq \dist^{\OL{\theta}}_H(u, v)$ necessarily holds as they have different parities. 
We also introduce the notion of parity $\rho(e)$ of each edge $e \in E_G$, which is defined as $\rho(e) = \Odd$ if $e$ is a
matching edge, and $\rho(e) = \Even$ otherwise. Intuitively, $\rho(e)$ represents the parity of the alternating path $P$ from $f$ such that $P \circ e$ can become an alternating path.

While a majority of notations and terminologies defined above actually depend on the matching 
$M$ currently given, $M$ is fixed in most of the following arguments, and thus we introduce them 
without describing an argument $M$ for simplicity. In addition, 
the subscripts $H$ of notations $\dist^{\theta}_H(u, v)$ and $\gamma_H(u, v)$ are also abbreviated 
if $H$ = $G$.


\subsection{CONGEST Model}
This paper assumes the CONGEST model, which is one of the standard models in
designing distributed graph algorithms. A distributed system is represented by a simple connected undirected graph 
$G$ of $n$ vertices and $m$ edges. Vertices and edges are uniquely identified by $O(\log n)$-bit integer values, and the 
computation follows round-based synchronous message passing. In one round, each vertex $v$ sends and receives $O(\log n)$-bit messages 
through each of its incident edges, and executes local computation following its internal state, local random bits, 
and all the received messages. A vertex can send different messages to different neighbors 
in one round, and all the messages sent in a round are necessarily delivered to their 
destinations within the same round. 


Throughout this paper, we assume that each vertex initially knows the value of $n$ and a 2-approximate upper bound $\AMmax$ of $\Mmax$. This assumption is not essential because both values are easily computed by an $O(\Mmax)$-round preprocessing presented in \cite{KI22}. 


\section{$O(\ell)$-Round Construction of Augmenting Path}
\label{sec:augpathconstruction}

\subsection{Specifications of Frameworks}
As stated in the introduction, the technical core of our algorithm is to develop an 
$O(\ell)$-round algorithm of constructing an augmenting path of length at most $\ell$ 
on the top of the Ahmadi--Kuhn framework. Such a construction is transformed into 
an $\tilde{O}(\Mmax)$-round maximum matching algorithm in the CONGEST model following the framework by the Hopcroft--Karp analysis. We first present the formal specification
of these two frameworks. On the Hopcroft--Karp framework, we have the following lemma:

\begin{lemma}[Ahmadi, Kuhn, and Oshman\cite{AKO18}]
\label{lma:HK-framework}
Assume a CONGEST algorithm $\SAPath(M, \ell)$ with the following property:
\begin{quote}
For any graph $G = (V, E)$ and matching $M \subseteq E$,
$\SAPath(M, \ell)$ finds a nonempty set of vertex-disjoint augmenting paths with high probability if $(G, M)$ has an augmenting path of length at most $\ell$, and it always terminates within $O(\ell)$ rounds. Each vertex $u \in V(G)$ outputs the predecessor and successor of the output augmenting path to which $u$ belongs (if it exists).
\end{quote}
Then there exists a randomized CONGEST algorithm computing a maximum matching within $\tilde{O}(\Mmax)$ rounds in total. 
\end{lemma}
The framework by Ahmadi and Kuhn is formalized by the following lemma:
\begin{lemma}[Ahmadi and Kuhn\cite{AK20}, Kitamura and Izumi\cite{KI22}]
\label{lma:AK-framework}
Assume that there exists an algorithm $\SAPath'(M, \ell)$ satisfying the specification of Lemma~\ref{lma:HK-framework} under the following additional assumptions.
\begin{itemize}
\item The graph $G$ contains exactly two unmatched vertices $f$ and $g$, which satisfies 
$\dist^{\Odd}(f, g) \leq \ell$.
\item Every vertex $v \in V(G)$ has a shortest alternating path from $f$ of length at most $\ell$, i.e., $v$ satisfies $\dist^{\gamma(f,v)}(f, v) \leq \ell$. In addition, $v$ knows $\dist^{\Odd}(f, v)$, $\dist^{\Even}(f,v)$, and the ID of $f$. This assumption 
implies the diameter of $G$ is $O(\ell)$. 
\end{itemize}
Then there exists a randomized algorithm $\SAPath(M, \ell)$ satisfying the specification of 
Lemma~\ref{lma:HK-framework} without these assumptions. 
\end{lemma}
While the proofs are given in the cited literatures, they do not state the lemmas
explicitly in the forms as above. We provide a complementary explanation of the proofs
in Appendix~\ref{appendix:frameworks}. In virtue of these lemmas, we can focus only on 
the construction of $\SAPath'$ of Lemma~\ref{lma:AK-framework}. Note that the algorithm 
$\SAPath'$ we present is deterministic, i.e., random bits are used only in the part of 
Ahmadi--Kuhn. Hence, we do not pay much attention to the success probability of the algorithm.

In the following argument, we assume that the input graph $G$ satisfies the assumption of
Lemma~\ref{lma:AK-framework}, and treat $\SAPath'$ as 
an algorithm of finding an augmenting path from $f$.
The notations $\dist^{\theta}(u, v)$ and 
$\gamma(u, v)$ are abbreviated to $\dist^{\theta}(v)$ and $\gamma(v)$ if $u = f$. 

\subsection{Alternating Base Tree and Edge Level}
\label{sec:abtreeandlevel}

Our algorithm uses the concept of alternating base trees proposed by Kitamura and Izumi~\cite{KI22}.

\begin{definition}[Alternating base tree]

An \emph{alternating base tree} for $G$, $M$, and $f$ is a rooted spanning tree $T$ of $G$
satisfying the following conditions:
\begin{itemize}
\item $f$ is the root of $T$.
\item For any $v \in V(G) \setminus \{f \}$ and
its parent $u$, $\dist^{\gamma(v)}(v) = \dist^{\OL{\gamma(v)}}(u) + 1$ holds.

\end{itemize}
\end{definition}
\begin{figure*}[tb]
\begin{center}
\includegraphics[width=13cm]{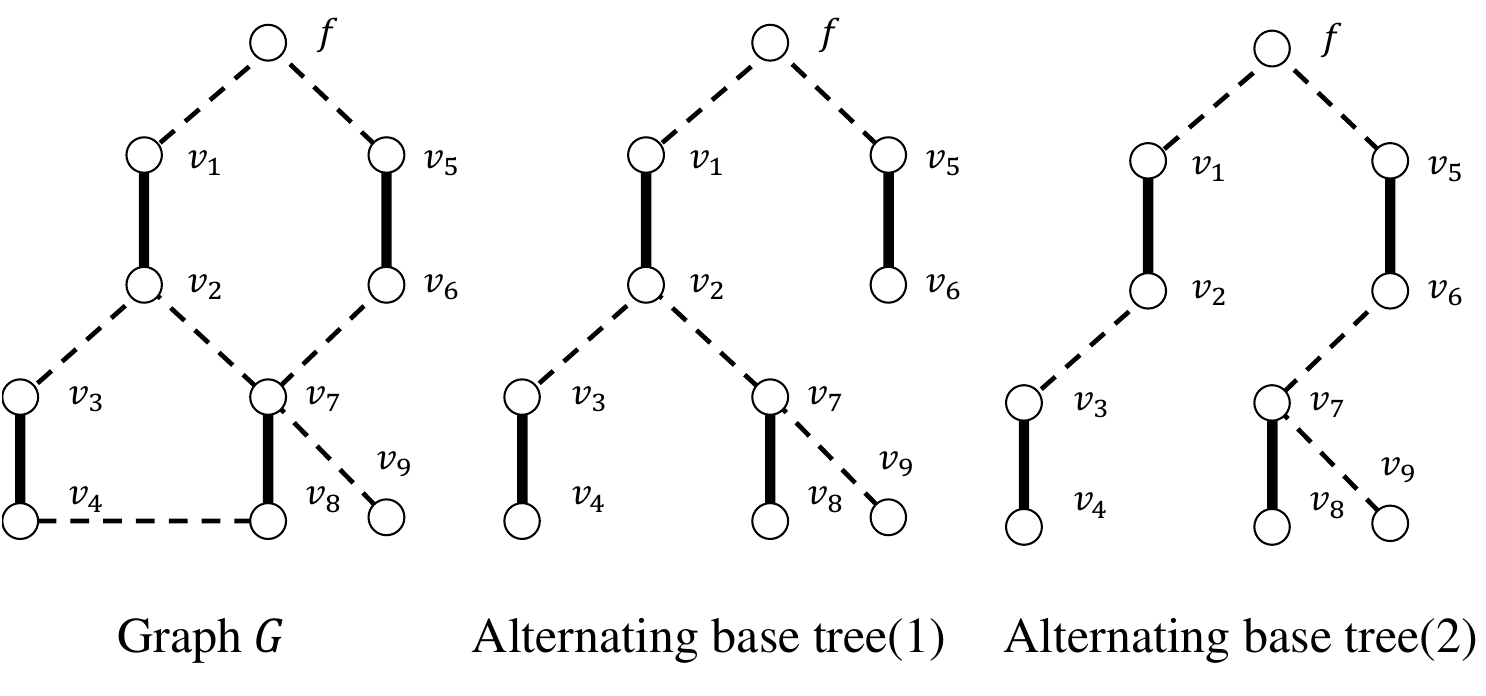}
\end{center}
\caption{Examples of alternating base trees (from \cite{KI22}). Solid lines are matching edges, and broken lines are non-matching edges.}
 \label{fig:abt}
\end{figure*}
It is easy to check that such a spanning tree always exists. As a vertex might have two or more shortest alternating paths, 
$T$ is not uniquely determined (see 
Figure~\ref{fig:abt} (1) and (2) for examples). In the following argument, we fix one alternating base tree $T$ arbitrarily chosen. The alternating base tree does not 
necessarily contain an alternating path from $f$ to each vertex $v$. For example, both alternating base trees in Figure~\ref{fig:abt} have no alternating path from $f$ to $v_{9}$.

We denote by $\parent(v)$ the parent of $v$ in $T$, and denote by $T_v$ the subtree of $T$ rooted by $v$. 
We define the \emph{level} of each edge as follows:
\begin{definition}[Edge level]
For any $e=\{u,v\}\in E(G)\setminus E(T)$, we define the \emph{sum-level} and \emph{max-level} of $e$ as $\sumlevel(e)=\dist^{\rho(e)}(u)+\dist^{\rho(e)}(v)$ and $\maxlevel(e)=\max(\dist^{\rho(e)}(u), \dist^{\rho(e)}(v))$, respectively.
The \emph{level} of an edge $e \in E(G)$, denoted by $\level(e)$, is defined 
as the ordered pair $(\sumlevel(e),\maxlevel(e))$ if $e \in E(G) \setminus E(T)$, and as $(0, 0)$ if $e \in E(T)$.
\end{definition}

We introduce the comparison operator $<$ of edge levels following their lexicographic order, 
i.e., for any two edges $e_1, e_2 \in E(G)$, $\level(e_1)<\level(e_2)$ holds if and only if either $\sumlevel(e_1)<\sumlevel(e_2)$ or $(\sumlevel(e_1)=\sumlevel(e_2)) \wedge (\maxlevel(e_1)<\maxlevel(e_2))$ holds.
Given an alternating path $P$, we refer to the maximum level of the edges in $P$ (with respect to $<$) as the \emph{level} of $P$,
which is denoted by $\level(P)$. The \emph{sum-level} of $P$ and its notation $\sumlevel(P)$ are also defined 
similarly.

The notion of edge levels is illustrated in Figure~~\ref{fig:edgelevel}.
A finite level of an edge $e$ implies that there exists an odd alternating walk starting and ending at $f$ and containing $e$.
Suppose that one of such walks forms a blossom, i.e., an odd cycle containing $e$ and an alternating path connecting $f$ and the odd cycle, called the \emph{stem}, that share exactly one vertex.
Then the sum-level of $e$ corresponds to the ``minimum volume'' of such a blossom, where the volume is measured 
by the number of edges in the blossom with double count of those in the stem. The max-level of $e$ is an upper bound for 
the ``height'' of a minimum-volume blossom $B$, which is defined as the distance from $f$ to the farthest vertex along $B$. 
Especially, the minimum max-level of edges in $B$ is equal to the height. 
If we take a minimum-volume blossom $B$ among all the blossoms containing $e$ in the cycle, for any vertex $w$ 
in the cycle except for the one shared with the stem, the two paths from $f$ to $w$ along $B$ are the shortest $\Odd$- and 
$\Even$-alternating paths to $w$ respectively.

\begin{figure*}[tb]
\begin{center}
\includegraphics[width=6cm]{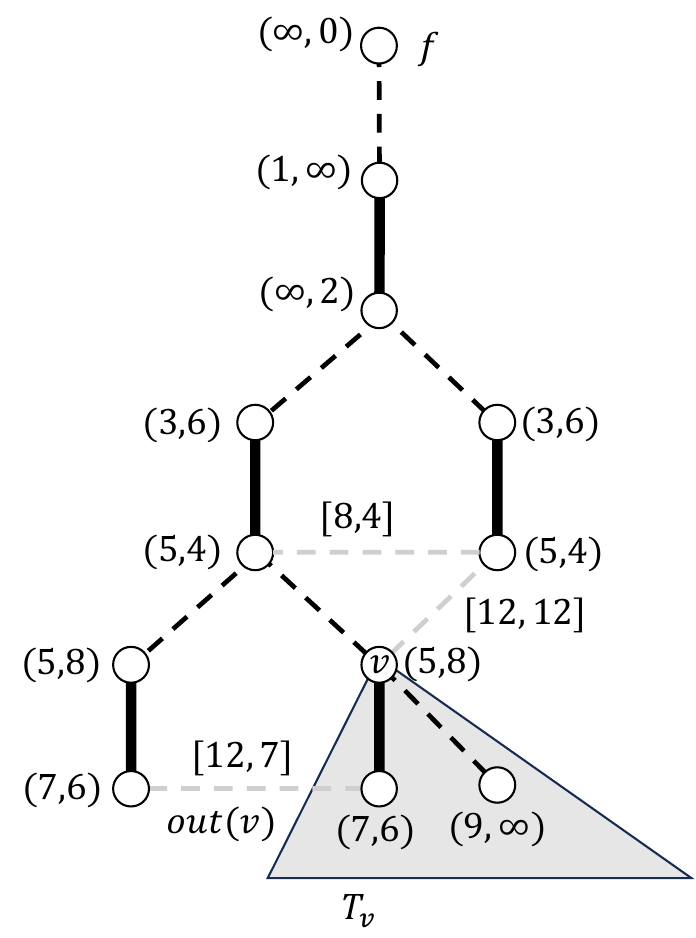}
\end{center}
\caption{An example of edge levels. Two values paired by parentheses respectively mean
the values of $\dist^{\Odd}(\cdot)$ and $\dist^{\Even}(\cdot)$. The values paired by brackets represent edge levels. All tree edges have level $[\infty, \infty]$, which is not explicitly described in the figure.}
 \label{fig:edgelevel}
\end{figure*}

\begin{definition}[Minimum outgoing edge (MOE) of $T_v$]
Let $v \in V(G)$ be a vertex. An edge which is not contained in $T$ (not restricted to $T_v$) and crosses 
$V(T_{v})$ and $V(G) \setminus V(T_v)$ is called an \emph{outgoing edge} of $T_v$. A \emph{minimum outgoing edge (MOE)} 
of $T_v$ is an outgoing edge of $T_v$ whose level is minimum. 
\end{definition}
For simplicity of the proof, for each vertex $v$ such that $T_v$ has no outgoing edge (including the case of $v = f$), 
we introduce a virtual outgoing edge $e_v$ of $T_v$ whose level is $(\infty,\infty)$.
While a minimum outgoing edge is not necessarily unique, we choose one of them for each $v \in V(G)$
as the \emph{canonical MOE} of $T_v$ following the rules below:
\begin{itemize}
\item If $T_v$ has no outgoing edge, the virtual edge $e_v$ is chosen.
\item If there exist MOEs incident to $v$, an arbitrary one of them is chosen.
\item Otherwise, an arbitrary MOE is chosen.
\end{itemize}
We denote the canonical MOE of $T_v$ by $\out(v)$, and $E^{\ast}(T_v)$ as the set of all outgoing 
edges of $T_v$ with a finite edge level. 
In the following argument, we denote the two endpoints of $\out(v)$ by $y_v$ and $z_v$, i.e., $\out(v) = \{y_v, z_v\}$. 
Using this notation, we always assume $z_v \in V(T_v)$. We also use the shorthand notation $\rho(v) = \rho(\out(v))$.


The key technical observation is the following lemma (the proof 
is deferred to Section~\ref{sec:correct}). 

\begin{lemma}
\label{lma:twopaths}
For any $t \in V(G)$ such that $E^{\ast}(T_t) \neq \emptyset$,
the following three statements hold (see Fig.~\ref{fig:lma3}):
\begin{itemize}
    \item[(S1)] There exists a shortest $\rho(t)$-alternating path $Y'$ from $f$ to $y_t$ of level less than $\level(\out(t))$.
    \item[(S2)] There exists a shortest $\rho(t)$-alternating path $Z'$ from $f$ to $z_t$ of level less than $\level(\out(t))$.
    \item[(S3)] The path $X' = Y' \circ \{y_t, z_t\} \circ \OL{Z'}[z_t, t]$ is a shortest $\OL{\gamma(t)}$-alternating path from $f$ to $t$.
    This does not depend on the choices of $Y'$ and $Z'$ in the above two statements.
\end{itemize}
\end{lemma}

An illustrative explanation of this lemma is provided in Figure~\ref{fig:lma3}.
Any bireachable vertex $t$ is contained in the odd cycle of some blossom $B$ ``strictly'' (i.e., $t$ is not the vertex shared with the stem of $B$). 
This cycle must contain an edge in $E^{\ast}(T_v)$ because the parent of $t$ is 
the predecessor of the shortest $\gamma(t)$-alternating path 
to $t$, and thus the $\OL{\gamma(t)}$-alternating path along $B$ must cross between $V(T_v)$ and $V(G) \setminus V(T_v)$ 
at least once. Lemma~\ref{lma:twopaths} implies that at least one blossom containing $\out(t)$ and $t$ crosses 
exactly once. In addition, since $\out(t)$ has the minimum sum-level of all edges in $E^{\ast}(T_v)$, some blossom 
containing $\out(t)$ is minimum of all ones containing $t$ (in their cycles strictly). 
Hence the $\gamma(t)$- and $\OL{\gamma(t)}$-alternating paths along such a blossom is shortest.

\begin{figure*}[tb]
\begin{center}
\includegraphics[width=6cm]{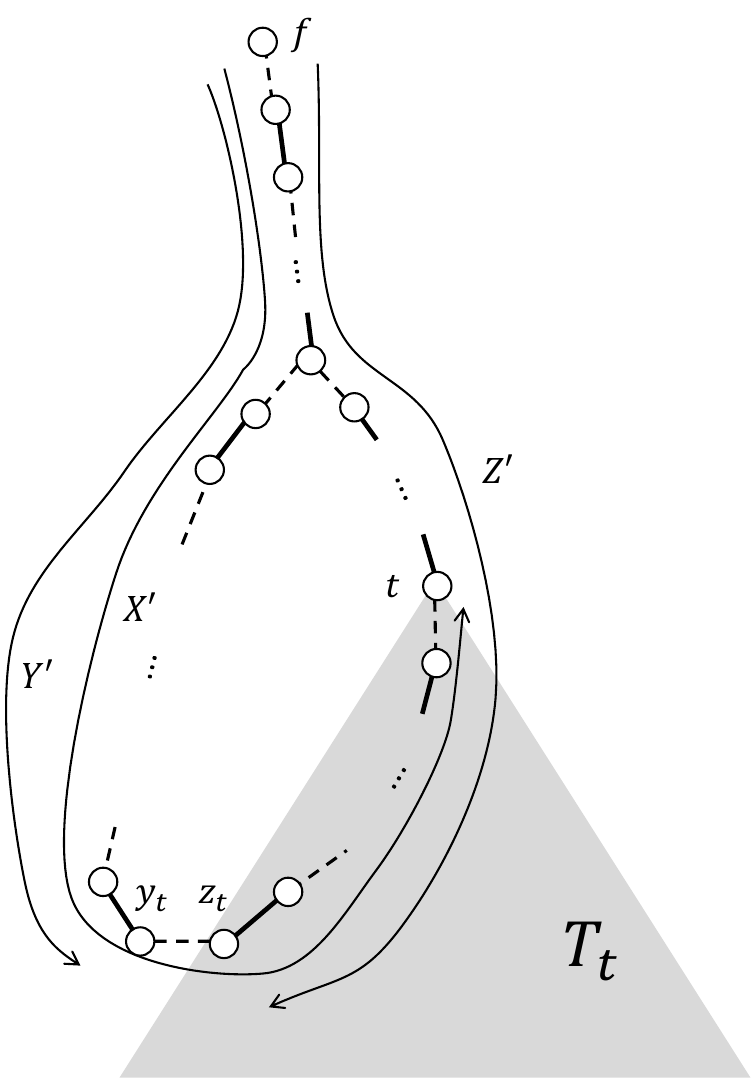}
\end{center}
\caption{A sketch of Lemma~\ref{lma:twopaths}.}
 \label{fig:lma3}
\end{figure*}

\subsection {Algorithm Details}
\label{subsec:algdetails}
We present our algorithm in the centralized manner, and discuss its distributed implementation 
in Section~\ref{sec:distributed}. We first introduce the notion of \emph{extendability}
of alternating paths. 

\begin{definition}
For $s, t \in V(G)$ such that $s$ is an ancestor of $t$, and $\theta \in \{\Odd, \Even\}$, 
a pair $(s, t)$ is called \emph{$\theta$-extendable} if there exists a shortest $\theta$-alternating
path $P$ from $f$ to $t$ satisfying $s \in V(P)$ and $\level(P) < \level(\out(s))$. 
Then the path $P$ satisfying the condition above is called a \emph{$\theta$-extension} of $(s, t)$, and $P[s, t]$ 
is called a \emph{$\theta$-extendable path} from $s$ to $t$. 
\end{definition}

A key feature of extendable paths is that the task of finding minimum-level extendable paths is self-reducible. 
The procedure 
$\FUNC(s, t, \theta)$ recursively computes a minimum-level $\theta$-extendable path from $s$ to $t$ 
following such a reduction scheme. 
The pseudocode of this procedure is presented in Algorithm~\ref{alg:func}. The principle behind the algorithm
is explained as follows:

\begin{itemize}
\item 
\sloppy{
For $\theta = \gamma(t)$: One can show that $(s, \parent(t))$ is 
$\OL{\theta}$-extendable, and hence the algorithm recursively invokes $\FUNC(s, \parent(t), \OL{\theta})$, which returns 
a minimum-level $\OL{\theta}$-extendable path $Y$ from $s$ to $\parent(t)$. It is easy to show that $Y$ does not contain $t$.
Thus, the algorithm finally returns $Y \circ \{\parent(t), t\}$, which is shown to be a minimum-level $\theta$-extendable path from $s$ to $t$. 
}

\item For $\theta = \OL{\gamma(t)}$: Since any $\OL{\gamma(t)}$-alternating path 
from $f$ to $t$ cannot contain the edge $\{\parent(t), t\}$, any $\theta$-extension of $(s, t)$
must contain an outgoing edge of $T_t$ to reach $t$. That is, $\level(\out(t)) < \level(\out(s))$ holds 
(recall that any $\theta$-extension of $(s, t)$ has a level less than $\level(\out(s))$). This means that both $y_t$ and $z_t$ are in $T_s$.
By Lemma~\ref{lma:twopaths}, 
there exist two alternating paths $Y' \circ \{y_t, z_t\}$ and $Z'$ from $f$ to $z_t$ whose levels 
are at most $\level(\out(t)) < \level(\out(s))$.
Since $\{\parent(s), s\}$ is the only edge of level less than $\level(\out(s))$ to reach the vertices in $V(T_s)$, 
both $Y'$ and $Z'$ necessarily contain $\{\parent(s), s\}$. 
Hence, $Y'$ and $Z'$ are respectively $\rho(t)$-extensions of $(s, y_t)$ and $(t, z_t)$, i.e., they are $\rho(t)$-extendable. 
The algorithm recursively invokes $\FUNC(s, y_t, \rho(t))$ and $\FUNC(t, z_t, \rho(t))$ for identifying two minimum-level 
$\rho(t)$-extendable paths, $Y$ from $s$ to $y_t$ and $Z$ from $t$ to $z_t$. Due to the existence of $Y'[s, y_t]$ and $Z'[t, z_t]$, 
whose levels are less than $\level(\out(t))$, 
the levels of $Y$ and $Z$ are also less than $\level(\out(t))$. Then $Y$ lies at the outside of $T_t$, and $Z$ lies at the inside 
of $T_t$. It implies that the concatenated path $Y \circ \{y_t, z_t\} \circ \OL{Z}$, which is the output of the algorithm, is 
an alternating path from $s$ to $t$. Utilizing the statement (S3) of Lemma~\ref{lma:twopaths}, one can prove that the output is 
a minimum-level $\theta$-extendable path.
\end{itemize}

\begin{algorithm}[ht] 
\caption{Procedure $\FUNC(s,t,\theta)$}
\label{alg:func}
{\setlength{\baselineskip}{14pt}
\begin{algorithmic}[1]
\IF {$s = t$}
\STATE return the path of length zero 
\ELSIF{$\theta=\gamma(t)$}
\STATE $Y \leftarrow \FUNC(s,\parent(t),\OL{\theta})$ 
\STATE return $Y \circ \{\parent(t), t\}$ 

\ELSE
\STATE $Y \leftarrow \FUNC(s,y_t,\rho(t))$ \COMMENT{ $\out(t) = \{y_t, z_t\}$, where $z_t$ is in $T_t$}
\STATE $Z \leftarrow \FUNC(t, z_t,\rho(t))$
\STATE return $Y \circ \{y_t,z_t\} \circ \OL{Z}$
\ENDIF
\end{algorithmic}
}
\end{algorithm}

The correctness of the algorithm follows the lemma below. The formal proof is presented in Section~\ref{sec:extpath}.

\begin{lemma}
\label{lma:augpathcorrectness}
For any tuple $(s, t, \theta)$ such that $(s, t)$ is $\theta$-extendable, $\FUNC(s, t, \theta)$ outputs a minimum-level $\theta$-extendable path from $s$ to $t$.
\end{lemma}

By the definition, any $\theta$-extendable path from $f$ to $t$ becomes 
a shortest $\theta$-alternating path from $f$ to $t$ (recall $\level(\out(f)) = (\infty, \infty)$).
Hence, it suffices to call $\FUNC(f, g, \Odd)$ for finding a shortest augmenting path from $f$ to $g$. 

\subsection{Distributed Implementation}
\label{sec:distributed}
%
In this section, we explain how $\FUNC$ is implemented in the CONGEST model with $O(\ell)$
rounds, where $\ell$ represents the length of the path returned by the procedure. We define the
distributed implementation of $\FUNC(s, t, \theta)$ as the CONGEST algorithm such that
its invocation is triggered by the vertex $t$, and finally each vertex in the output path 
knows its successor and predecessor. As precomputed information, our distributed implementation 
assumes that each vertex $u$ locally owns the following knowledge:
\begin{itemize}
    \item The parent and children of $u$ in $T$, and the value of $\gamma(u)$.
    \item The IDs of two endpoints of $\out(u)$ (if $\out(u)$ is not virtual).  
    \item Let $R_{v, z_v}$ be the unique path from $v$ to $z_v$ in $T$.
    If $u \in V(R_{v, z_v})$, $u$ has the information for routing messages from 
    $v$ to $z_v$. More precisely, $u$ knows the pair $(z_v, u')$, where $u'$ is the successor of $u$ in 
    $R_{v, z_v}$.
\end{itemize}
The knowledge above is obtained by the almost same technique as that presented in~\cite{KI22},
whose detail is explained in Appendix~\ref{appendix:aggregate}. Utilizing the precomputed information stated above, it is very straightforward to implement $\FUNC$ in the CONGEST model.
If $\theta = \gamma(t)$, $t$ sends the triggering message with the information of $s$ and $\theta$ to $\parent(t)$ for invoking $\FUNC(s, \parent(t), \OL{\theta})$, as well as 
adding $\{\parent(t), t\}$ to the output edge set. Otherwise, $t$ sends the triggering message with $s$ to the endpoints of $\out(t)$ along the path $R_{t, z_t} \circ \out(t)$. Receiving the 
message, $y_t$ and $z_t$ add $\out(t)$ to the output edge set, and respectively invoke $\FUNC(s, y_t, \rho(t))$ and $\FUNC(t, z_t, \rho(t))$. Notice that $\rho(t)$ is local information available at $y_t$ and $z_t$.

The correctness of this distributed implementation obviously follows Lemma~\ref{lma:augpathcorrectness}.
The remaining issue is the running-time analysis. The running time of each invocation is dominated
by the rounds for transferring triggering messages. Particularly, the case of $\theta = \OL{\gamma(t)}$
is the only nontrivial point. One can show that the length of $|R_{t, z_t}|$ is bounded
by the minimum of the lengths of extendable paths from $s$ to $y_t$ and from $t$ to $z_t$. The
following lemma holds.

\begin{lemma}
 \label{lma:token_time}
 For any $\overline{\gamma(t)}$-extendable pair $(s,t)$, let $Z$ be a $\rho(t)$-extendable path from $t$ to $z_t$ and $Y$ be a $\rho(t)$-extendable path from $s$ to $y_t$. Then, $|R_{t, z_t}| \leq |Z|$ and $|R_{t, z_t}| \leq |Y|$ hold.
\end{lemma}
Thus, the total running time is bounded as follows:
\begin{lemma}
\label{lma:exe_time}
    Let $g(h)$ be the worst-case round complexity of $\FUNC(s,t,\theta)$ over all 
    $(s, t, \theta)$ admitting a $\theta$-extendable path from $s$ to $t$ of length $h$.
    For any $0 \leq h \leq \ell$, $g(h) \leq 2h$ holds.
\end{lemma}
As we stated in Section~\ref{subsec:algdetails}, $\FUNC(f, g, \Odd)$ suffices to find a
shortest augmenting path from $f$ to $g$. For its length $\ell$, the running time is 
at most $2\ell$. Putting all the results
presented in this section together, we obtain the following corollary.

\begin{corollary}
\label{corol:augpathalgorithm}
There exists an algorithm $\SAPath'(M, \ell)$ satisfying the specification of Lemma~\ref{lma:AK-framework}.
\end{corollary}

Combining with Lemmas~\ref{lma:HK-framework} and \ref{lma:AK-framework}, Theorem~\ref{thm:main} is deduced.

\section{Detailed Proofs}
\label{sec:correct}

\subsection{Fundamental Properties of Alternating Base Trees} 

Before proceeding to the proofs of key technical lemmas,
we present several fundamental properties related to alternating base trees and edge levels.
We say that \emph{a path $P$ from $s$ to $t$ lies at the inside of $T_v$} if 
$V(P) \subseteq V(T_v)$ holds. We also say that \emph{$P$ lies at the outside of $T_v$} if $V(P) \subseteq (V(G) \setminus V(T_v)) \cup \{v\}$, 
i.e., $P$ is either completely disjoint from $V(T_v)$ or disjoint from $V(T_v)$ except for one of its endpoints that is $v$.

\begin{lemma} 
\label{lma:outgoingedgeproperty}
The following statements hold:
\begin{itemize}
    \item[(S1)] For any alternating path $P$ from $f$ to $t$ and an ancestor $s$ of $t$ (including the case of $s = t$), 
    $P$ contains $\{\parent(s), s\}$ if $\level(P) < \level(\out(s))$. Hence any alternating path $P$ from $f$ to $t$ not 
    containing $\{\parent(t), t\}$ satisfies $\level(P) \geq \level(\out(t))$.
    \item[(S2)] For any $t \in V(G)$ and a shortest $\OL{\gamma(t)}$-alternating path $P$ from $f$ to 
    $\parent(t)$, $P$ does not contain $t$. 
    \item[(S3)] The parity of any alternating path $P$ from $f$ to $t$ terminating with edge $\{\parent(t), t\}$ is $\gamma(t)$,
    and thus any $\OL{\gamma(t)}$-alternating path from $f$ to $t$ does not contain $\{\parent(t), t\}$.
    \item[(S4)] For $u \in V(G)$, $\dist^{\gamma(\parent(v))}(\parent(v)) 
< \dist^{\gamma(v)}(v)$ holds, and thus for any descendant $v$ of $u$, $\dist^{\gamma(u)}(u) < \dist^{\gamma(v)}(v)$ holds.
    \item[(S5)] Any shortest $\gamma(t)$-alternating $P$ path from $f$ to $t$ lies at the outside of $T_t$.
\end{itemize}
\end{lemma}

\begin{proof}
\textbf{Proof of (S1)}: The path $P$ must contain at least 
one edge in $E^{\ast}(T_s) \cup \{ \{\parent(s), s\}\}$ to reach $s$. Since only $\{\parent(s), s\}$ has a level less than $\level(\out(s))$, $P$ necessarily contains it.

\noindent
\textbf{Proof of (S2)}: Suppose for contradiction that there exists a shortest $\OL{\gamma(t)}$-alternating path $P$ from 
$f$ to $\parent(t)$ that contains $t$. Then $P[f, t]$ is an alternating path from $f$ to $t$ of length at most $|P| - 1 < \dist^{\OL{\gamma(t)}}(f, \parent(t))$.
However, by the definition of alternating base trees, the shortest alternating path from $f$ to $t$ is of length 
$\dist^{\OL{\gamma(t)}}_{G}(f,\parent(t))+1$. It is a contradiction.

\noindent
\textbf{Proof of (S3)}: Since any two alternating paths from $f$ terminating with the same edge must have 
the same parity, it suffices to show that there exists a $\gamma(t)$-alternating path from $f$ terminating
with $\{\parent(t), t\}$. Let $Q$ be any shortest $\OL{\gamma(t)}$-alternating path from $f$ to $\parent(t)$. 
By statement (S2), $Q$ does not contain $t$, and thus we obtain a $\gamma(t)$-alternating path 
$P' = Q \circ \{\parent(t), t\}$ from $f$ to $t$.

\noindent
\textbf{Proof of (S4)}:
By the definition of alternating base trees, $\dist^{\gamma(v)}(v) = 
\dist^{\OL{\gamma(v)}}(\parent(v)) + 1$ holds. 
Since the case of $\OL{\gamma(v)} = \gamma(\parent(v))$ is obvious, we consider the case of 
$\gamma(v) = \gamma(\parent(v))$. By definition, we obtain $\dist^{\gamma(\parent(v))}(\parent(v)) < 
\dist^{\OL{\gamma(\parent(v))}}(\parent(v))$. It follows $\dist^{\gamma(v)}(v) = 
\dist^{\OL{\gamma(v)}}(\parent(v)) + 1 = \dist^{\OL{\gamma(\parent(v))}}(\parent(v)) + 1 > \dist^{\gamma(\parent(v))}(\parent(v))$. 

\noindent
\textbf{Proof of (S5)}:
Every vertex $u \in V(P)$ satisfies $\dist^{\gamma(u)}(u) \leq \dist^{\gamma(t)}(t)$. By statement (S4), $P$ 
does not contain any vertex in $V(T_t) \setminus \{t \}$, i.e., $P$ lies at the outside of $T_s$. 
\end{proof}

\subsection{Proof of Lemma~\ref{lma:twopaths}}

\begin{figure*}[tb]
\begin{center}
\includegraphics[width=12cm]{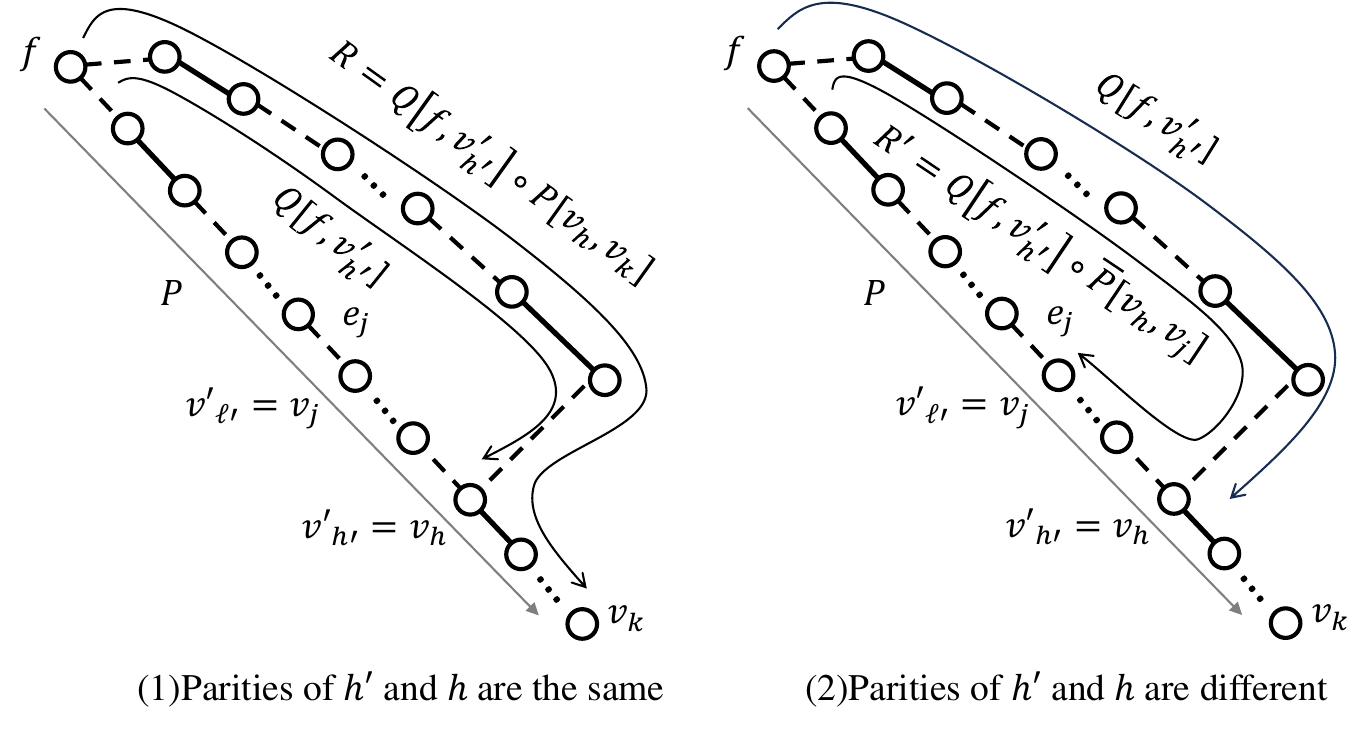}
\end{center}
\caption{Illustration of Case 2 in the proof of Lemma~\ref{lma:levelbound}.}
 \label{fig:lma8}
\end{figure*}

\begin{lemma}
\label{lma:levelbound}
For any $\theta\in\{\Odd, \Even\}$ and $\theta$-reachable $u \in V(G)$, there exists a shortest 
$\theta$-alternating path $P$ from $f$ to $u$ such that $\sumlevel(P) \leq 2\dist^{\theta}(u) - 1$.
\end{lemma}

\begin{proof}
Let $\mathcal{Q}_k$ be the set of pairs $(v, \theta)$ such that $\dist^{\theta}(v) = k$.
We prove that the lemma holds for any $k$ and $(u, \theta) \in \mathcal{Q}_k$ by
the induction on $k$.
(Basis) For any $(u, \theta) \in \mathcal{Q}_1$, the edge $\{f,u\}$ is in the alternating
base tree. Since the sum-level of any tree edge is zero, the lemma obviously holds.
(Inductive Step) Suppose as the induction hypothesis that the lemma holds for any $(u', \theta') \in \bigcup_{k' < k} \mathcal{Q}_{k'}$, and consider the pair $(u, \theta) \in \mathcal{Q}_k$. 
We divide the proof into the following two cases:

\noindent
\textbf{(Case 1)} $\theta=\gamma(u)$:
Let $P$ be any minimum-level shortest $\OL{\gamma(u)}$-alternating path from $f$ to $\parent(u)$. By 
Lemma~\ref{lma:outgoingedgeproperty}(S2), $P$ does not contain $u$, and thus 
we obtain an alternating path $P' = P \circ (\parent(u), u)$, whose parity is $\gamma(u)$
by Lemma~\ref{lma:outgoingedgeproperty}(S3). 
By the definition of alternating base trees, the length of $P'$ is $|P| + 1 = 
\dist^{\OL{\gamma(u)}}(\parent(u)) + 1 = \dist^{\gamma(u)}(u) = k$, i.e., $(\parent(u), \OL{\gamma(u)}) \in \mathcal{Q}_{k-1}$.
By the induction hypothesis, we obtain $\sumlevel(P) \leq 2k - 3$. It follows 
$\sumlevel(P') = \max\{\sumlevel(P), \sumlevel(\{\parent(u), u\})\} \leq 2k - 3$ 
(recall $\sumlevel(\{\parent(u), u\}) = 0$). That is, the sum-level of $P'$ is at most 
$2\dist^{\gamma(u)}(u) - 1$. Since $|P'| = \dist^{\gamma(u)}(u)$ holds, it is the path satisfying the condition of the lemma.

\noindent
\textbf{(Case 2)} $\theta=\OL{\gamma(u)}$: 
Let $P={v_0, e_1, v_1, \dots, e_{k}, v_{k}}$ be any shortest $\theta$-alternating path from $f$ 
to $u$, where $v_0 = f$ and $v_k = u$, and $e_j$ ($1 \leq j \leq k$) be the edge such that 
$\sumlevel(e_j) > 2k - 1$ and $\sumlevel(P[v_j, v_k]) \leq 2k - 1$ hold. If such an edge does not exist, 
$P$ obviously satisfies $\sumlevel(P) \leq 2\dist^{\theta}(u) - 1$, and the proof finishes. 
Since $P[f, v_{j-1}]$ is a $\rho(e_j)$-alternating path from $f$ to $v_{j-1}$,
we obtain $\dist^{\rho(e_j)}(v_{j-1}) \leq |P[f, v_{j-1}]| = j - 1$. Then $\dist^{\rho(e_j)}(v_j) = 
\sumlevel(e_j) - \dist^{\rho(e_j)}(v_{j-1}) > 2k - j \geq k$ holds. Since $P[f, v_{j}]$ is a $\OL{\rho(e_j)}$-alternating 
path from $f$ to $v_j$ of length $j \leq k$, we have $\dist^{\rho(e_j)}(v_j) > k \geq \dist^{\OL{\rho(e_j)}}(v_j)$. That is, 
$\OL{\rho(e_j)} = \gamma(v_j)$ holds. It follows $j < k$ since we assume that the parity $\theta$ 
of $P = P[f, v_k]$ is $\OL{\gamma(v_k)}$.

The induction hypothesis provides a shortest $\gamma(v_j)$-alternating path $Q = v'_0, e'_1, v'_1, \dots, 
e'_{\ell}, v'_\ell$ (where $\ell = \dist^{\gamma(v_j)}(v_j)$, $v'_0 = f$, and $v'_\ell = v_j$) from $f$ to $v_j$ 
of sum-level at most $2\ell - 1$. Since $Q$ is a shortest alternating path from $f$ to $v_j$, its length $\ell$ 
is at most $|P[f, v_j]| = j$.
Let $h'$ be the smallest index such that $v'_{h'} \in V(P[v_j, v_k])$, and $h$ be the index such that 
$v'_{h'} = v_h$ holds ($1 \leq h' \leq \ell$, $j \leq h \leq k$). Note that such $h'$ and $h$ necessarily exist because $v_j$ is contained in both $V(Q)$ and $V(P[v_j, v_k])$.

Consider the case that the parities of $h$ and $h'$ are the same (see Fig.~\ref{fig:lma8}(1)). Since $Q[f, v'_{h'}]$ and $P[v_h, v_k]$ are vertex disjoint except 
for $v'_{h'} = v_h$, we obtain a $\theta$-alternating path $R = Q[f, v'_{h'}] \circ P[v_h, v_k]$ from $f$ to $v_k = u$.
The length of $R$ is $h' + k - h \leq \ell + k - j \leq k = \dist^{\theta}(u)$, i.e., $R$ is a 
shortest $\theta$-alternating path. In addition, $\sumlevel(Q) \leq 2j - 1 \leq 2k - 1$ holds and $P[v_j, v_k]$ 
does not contain any edge of sum-level greater than $2k - 1$. It implies $\sumlevel(R) = \max\{\sumlevel(Q[f, v'_{h'}]), 
\sumlevel(P[v_h, v_k])) \leq 2k - 1$. That is, $R$ satisfies the condition of the lemma.

Finally, we rule out the case that the parities of 
$h$ and $h'$ are different (see Fig.~\ref{fig:lma8}(2)). Suppose for contradiction that such a case occurs. Then $R'=Q[f, v'_{h'}] \circ 
\OL{P}[v_h, v_j]$ becomes a $\rho(e_j)$-alternating path from $f$ to $v_j$ of length $h' + h - j \leq 
j + k - j = k$ (because of $h' \leq \ell \leq j$ and $|\OL{P}[v_h, v_j]| \leq |P[v_j, v_k]|$). 
Since $\dist^{\rho(e_j)}(v_{j-1}) \leq |P[f, v_{j-1}]| \leq k - 1$ holds, we have 
$\sumlevel(e_j) \leq 2k - 1$ holds. It is a contradiction.
\end{proof}

\begin{lemma} 
\label{lma:Y2}
For any $t \in V(G)$ such that $E^{\ast}(T_t)$ is not empty, the following two inequalities hold:
\begin{itemize}
    \item[(S1)] $\level(\out(t)) \leq (\dist^{\Odd}(t) + \dist^{\Even}(t) - 1, \dist^{\OL{\gamma(t)}}(t) - 1)$.    
    \item[(S2)] $\dist^{\OL{\gamma(t)}}(t) \geq \dist^{\rho(t)}(z_t) + 
\dist^{\rho(t)}(y_t) + 1 - \dist^{\gamma(t)}(t)$.
\end{itemize}
\end{lemma}
\begin{proof}
Let $P$ be any shortest $\gamma(t)$-alternating path from $f$ to $t$. 
By Lemma~\ref{lma:outgoingedgeproperty}(S5), $P$ lies at the outside of $T_t$.
Let $Q$ be any shortest $\OL{\gamma(t)}$-alternating path from $f$ to $t$. 
By Lemma~\ref{lma:outgoingedgeproperty}(S3), $Q$ does not contain $\{\parent(t), t\}$ and thus 
contains at least one edge in $E^{\ast}(T_t)$. 
Let $e = \{y',z'\}$, where $z' \in V(T_t)$, be the edge in $E(Q) \cap E^{\ast}(T_t)$ such that $Q[z', t]$ 
does not contain any edge in $E^{\ast}(T_t)$. Since $\OL{Q}[t, z']$ lies at the inside of $T_t$, 
$P \circ \OL{Q}[t, z']$ is a $\rho(e)$-alternating path from $f$ to $z'$ and thus
$\dist^{\rho(e)}(z') \leq |P \circ \OL{Q}[t, z']|$ holds. 
In addition, since $Q[f, y']$ is a  $\rho(e)$-alternating path from $f$ to $y'$, 
$\dist^{\rho(e)}(y') \leq |Q[f, y']|$ also holds.

\noindent
\textbf{Proof of (S1)}: Since $\level(\out(t)) \leq \level(e)$ holds, it suffices 
to show  $\level(e) \leq (\dist^{\Odd}(t) + \dist^{\Even}(t) - 1, 
\dist^{\OL{\gamma(t)}}(t) - 1)$. We have
\begin{align*}
    \sumlevel(e)&=\dist^{\rho(e)}(z')+\dist^{\rho(e)}(y')\\
    &\leq |P \circ \OL{Q}[t, z']| + |Q[f, y']|\\
    &= |P| + |\OL{Q}[t, z']| +  |Q[f, y']| \\ 
    &= |P| + |Q| - 1 \\
    &= \dist^{\gamma(t)}(t) + \dist^{\OL{\gamma(t)}}(t) - 1 \\
    &= \dist^{\Odd}(t) + \dist^{\Even}(t) - 1.
\end{align*}

Next, we upper bound $\maxlevel(e) = \max\{\dist^{\rho(e)}(y'), \dist^{\rho(e)}(z')\}$ 
by $\dist^{\OL{\gamma(t)}}(t) - 1$. We have $\dist^{\rho(e)}(y') \leq |Q[f, y']| \leq |Q| - 1 = 
\dist^{\OL{\gamma(t)}}(t) - 1$.
Since $z'$ is a descendant of $t$ (including $t$ itself), $|Q[f, z']| > \dist^{\gamma(t)}(t)$ holds as follows.
If $z' = t$, then $|Q[f, z']| = |Q| > |P| = \dist^{\gamma(t)}(t)$; otherwise, $|Q[f, z']| \geq \dist^{\gamma(z')}(z') > \dist^{\gamma(t)}(t)$ by Lemma~\ref{lma:outgoingedgeproperty}(S4). 
Then $\dist^{\rho(e)}(z') \leq |P \circ \OL{Q}[t, z']| = \dist^{\gamma(t)}(t) + (\dist^{\OL{\gamma(t)}}(t) - |Q[f, z']|) 
< \dist^{\OL{\gamma(t)}}(t)$ holds. Consequently, we have 
$\maxlevel(e) \leq \dist^{\OL{\gamma(t)}}(t) - 1$.

\noindent
\textbf{Proof of (S2)}: Since $P \circ \OL{Q}[t, z']$ is a $\rho(e)$-alternating path from $f$ to $z'$, 
$\dist^{\rho(e)}(z') \leq |P \circ \OL{Q}[t, z']| =|Q[z',t]| + \dist^{\gamma(t)}(t)$ holds, i.e., 
$|Q[z',t]| \geq \dist^{\rho(e)}(z') - \dist^{\gamma(t)}(t)$ holds. Then we have 
\begin{align*}
    \dist^{\OL{\gamma(t)}}(t) & =    |Q| \\
    &  =  |Q[f, y'] \circ e \circ Q[z', t]|\\
    &  \geq  \dist^{\rho(e)}(y') + 1 + \dist^{\rho(e)}(z') - \dist^{\gamma(t)}(t) \\
    &  \geq  \dist^{\rho(t)}(y_t) + 1 + \dist^{\rho(t)}(z_t) - \dist^{\gamma(t)}(t),
\end{align*}
where the last inequality follows from the definition of MOEs, i.e.,
$\sumlevel(e) \geq \sumlevel(\out(t))$ implies that
$\dist^{\rho(e)}(z')+ \dist^{\rho(e)}(y') \geq \dist^{\rho(t)}(z_t) + \dist^{\rho(t)}(y_t)$. 
\end{proof}


\begin{lemma}
\label{lma:Y6}
For any $t \in V(G)$ such that $E^{\ast}(T_t)$ is not empty, if $\dist^{\rho(t)}(y_t) \leq \dist^{\rho(t)}(z_t)$ holds,
$\dist^{\OL{\rho(t)}}(z_t) \leq \dist^{\rho(t)}(y_t) + 1$ holds.
\end{lemma}

\begin{proof}
Utilizing Lemma~\ref{lma:levelbound} with $\theta = \rho(t)$ and $u = y_t$, there exists a shortest 
$\rho(t)$-alternating path $P$ from $f$ to $y_t$ such that $\sumlevel(P) \leq 2\dist^{\rho(t)}(y_t) - 1$ 
(note that $y_t$ is $\rho(t)$-reachable because any edge $e \in E^{\ast}(T_t)$ has a finite edge level and thus
$\dist^{\rho(t)}(y_t)$ is finite). Since we assume $\dist^{\rho(t)}(y_t) \leq \dist^{\rho(t)}(z_t)$ 
we also have $\sumlevel(\out(t)) = \dist^{\rho(t)}(y_t) + \dist^{\rho(t)}(z_t) 
\geq 2\dist^{\rho(t)}(y_t)$, i.e., no outgoing edge of $T_t$ is contained in $P$.
It implies that $P$ lies at the outside of $T_t$ and thus does not contain $z_t$. 
Hence $P' = P \circ \{y_t, z_t\}$ is a $\OL{\rho(t)}$-alternating path from $f$ to $z_t$, whose length 
is $\dist^{\rho(t)}(y_t) + 1$. That is, $\dist^{\OL{\rho(t)}}(z_t) \leq \dist^{\rho(t)}(y_t) + 1$ holds.
\end{proof}



\begin{lemma}
\label{lma:MOEparity}
For any $t, t' \in V(G)$ such that $E^{\ast}(T_t)$ and $E^{\ast}(T_{t'})$ are not empty,
if $\level(\out(t)) = \level(\out(t'))$ holds, then $\rho(t) = \rho(t')$ holds.
\end{lemma}

\begin{proof}
In particular, we have $\maxlevel(\out(t)) = \maxlevel(\out(t'))$, i.e., $\max\{\dist^{\rho(t)}(y_{t}), \dist^{\rho(t)}(z_{t})\} = \max\{\dist^{\rho(t')}(y_{t'}), \dist^{\rho(t')}(z_{t'})\}$.
The both sides must have the same parity, which means $\rho(t) = \rho(t')$.
\end{proof}



\let\temp\thelemma
\renewcommand{\thelemma}{\ref*{lma:twopaths}}
\begin{lemma}[Repeat]
For any $t \in V(G)$ such that $E^{\ast}(T_t) \neq \emptyset$,
the following three statements hold (see Fig.~\ref{fig:lma3}):
\begin{itemize}
    \item[(S1)] There exists a shortest $\rho(t)$-alternating path $Y'$ from $f$ to $y_t$ of level less than $\level(\out(t))$.
    \item[(S2)] There exists a shortest $\rho(t)$-alternating path $Z'$ from $f$ to $z_t$ of level less than $\level(\out(t))$.
    \item[(S3)] The path $X' = Y' \circ \{y_t, z_t\} \circ \OL{Z'}[z_t, t]$ is a shortest $\OL{\gamma(t)}$-alternating path from $f$ to $t$.
    This does not depend on the choices of $Y'$ and $Z'$ in the above two statements.
\end{itemize}
\end{lemma}
\let\thelemma\temp
\addtocounter{lemma}{-1}

\begin{proof}
For $k \in \mathbb{N}^2$, let $\mathcal{V}_k$ be the set of the vertices $v$ such that $E^{\ast}(T_v)$ is not empty and 
$\level(\out(v)) = k$ holds. We show that the lemma holds for any $k$ and $t \in \mathcal{V}_k$ 
following the induction on
$k$. The base case of $k = (0, 0)$ is obvious (as $\mathcal{V}_k = \emptyset$), and thus we focus on the inductive step. The proof of the inductive step 
consists of the following three parts: 
\begin{itemize}
\item[(P1)] Supposing as the induction hypothesis that the lemma holds for any $t' \in \bigcup_{k' < k} \mathcal{V}_{k'}$, 
show the statement (S2) holds for any vertex $t \in \mathcal{V}_k$.
\item[(P2)] Supposing that the lemma holds for any $t' \in \bigcup_{k' < k} \mathcal{V}_{k'}$ as the induction hypothesis, and
that the statement (S2) holds for any $t \in \mathcal{V}_k$ by the consequence of (P1), show the statement $(S1)$ holds for any 
$t \in \mathcal{V}_k$.
\item[(P3)] For any $t \in \mathcal{V}_k$, (S1) and (S2) imply (S3).
\end{itemize}

\noindent
\textbf{The proof of (P1)}: The proof is divided into the following three cases:

\noindent
\textbf{(Case 1)} $\dist^{\rho(t)}(z_t) \leq \dist^{\rho(t)}(y_t)$: Then $\sumlevel(\out(t)) \geq 2\dist^{\rho(t)}(z_t)$ holds. 
By Lemma~\ref{lma:levelbound}, there exists a shortest $\rho(t)$-alternating path $P$ from $f$ to $z_t$ of sum-level 
at most $2\dist^{\rho(t)}(z_t)  - 1 < \sumlevel(\out(t))$, i.e., $P$ is the path satisfying the condition of (S2).

\noindent
\textbf{(Case 2)} $\dist^{\rho(t)}(z_t) > \dist^{\rho(t)}(y_t)$ and 
$\level(\out(z_t))<\level(\out(t))$: By the induction hypothesis, there exists a shortest $\OL{\gamma(z_t)}$-alternating path $X'$ from 
$f$ to $z_t$ whose level is bounded by $\level(\out(z_t)) < \level(\out(t))$.
Since $\dist^{\rho(t)}(z_t) > \dist^{\rho(t)}(y_t)$, by Lemma~\ref{lma:Y6}, we obtain $\dist^{\OL{\rho(t)}}(z_t) \leq \dist^{\rho(t)}(y_t) + 1 \leq \dist^{\rho(t)}(z_t)$.
It implies $\OL{\rho(t)} = \gamma(z_t)$ because of $\dist^{\OL{\rho(t)}}(z_t) \neq \dist^{\rho(t)}(z_t)$.
Then $X'$ is the path satisfying the condition of (S1).

\noindent
\textbf{(Case 3)} $\dist^{\rho(t)}(z_t) > \dist^{\rho(t)}(y_t)$ and $\level(\out(z_t)) \geq \level(\out(t))$: 
We show that this case never happens by leading to a contradiction. 
Let $t' = z_t$, $y_{t'} = y_{z_t}$, and $z_{t'} = z_{z_t}$ for short.
Since $\out(t)$ is an edge incident to $z_t$ which reaches the vertex $y_t$ at the outside of $T_t$, 
$\out(t)$ is also an outgoing edge of $T_{z_t}$, i.e., $\level(\out(z_t)) \leq \level(\out(t))$ holds. 
Combining with the assumption of $\level(\out(z_t)) \geq \level(\out(t))$, we obtain 
$\level(\out(z_t)) = \level(\out(t))$. In addition, 
by the rule of choosing canonical MOEs, the edges of level $\level(\out(z_t))$ incident to $z_t$ 
have higher priority. Hence $\out(z_t)$ is incident to $z_t$, i.e., $z_{t'} = t'$ holds. 
Since $\rho(t) = \rho(t')$ is obtained by Lemma~\ref{lma:MOEparity}, the assumption of 
$\dist^{\rho(t)}(z_t) > \dist^{\rho(t)}(y_t)$ implies $\dist^{\rho(t')}(z_{t'}) > \dist^{\rho(t')}(y_{t'})$. 
Then we obtain $\dist^{\OL{\rho(t')}}(z_{t'}) \leq \dist^{\rho(t')}(y_{t'}) + 1$ by Lemma~\ref{lma:Y6}.
Putting all the (in)equalities above together, we conclude $\dist^{\OL{\rho(t')}}(t') \leq \dist^{\rho(t')}(y_{t'}) + 1 \le \dist^{\rho(t')}(t')$. 
Then, the level of $\out(t')$ is bounded as follows:
\begin{align*}
\sumlevel(\out(t'))&= \dist^{\rho(t')}(y_{t'}) + \dist^{\rho(t')}(z_{t'})\\
&\geq \dist^{\OL{\rho(t')}}(t') -1  + \dist^{\rho(t')}(t') \\
&= \dist^{\Odd}(t') + \dist^{\Even}(t') -1, \\
\maxlevel (\out(t')) &= \max\{\dist^{\rho(t')}(y_{t'}), \dist^{\rho(t')}(z_{t'})\}\\
&= \dist^{\rho(t')}(t') \\
&> \dist^{\OL{\gamma(t')}}(t') - 1.
\end{align*}
That is, $\level(\out(t')) > (\dist^{\Odd}(t') + \dist^{\Even}(t') -1, \dist^{\OL{\gamma(t')}}(t') - 1)$. 
It contradicts Lemma~\ref{lma:Y2}(S1).

\noindent
\textbf{The proof of (P2)}: The proof is divided into the two cases below:

\noindent
\textbf{(Case 1)} $\level(\out(t)) \leq \level(\out(y_t))$:
Suppose that $y_t$ is an ancestor of $z_t$.
If $\rho(t) = \gamma(y_t)$, then a shortest $\rho(t)$-alternating path from $f$ to $y_t$ plus $\out(t)$ is a $\rho(t)$-alternating path from $f$ to $z_t$, say $P$.
By Lemma~\ref{lma:outgoingedgeproperty}(S4), we obtain $\dist^{\gamma(z_t)}(z_t) \le \dist^{\rho(t)}(z_t) \le |P| = \dist^{\gamma(y_t)}(y_t) + 1 \le \dist^{\gamma(z_t)}(z_t)$, which implies $\rho(t) = \gamma(z_t)$.
However, this implies that $\dist^{\gamma(y_t)}(y_t)$ and $\dist^{\gamma(z_t)}(z_t)$ have the same parity, so $\dist^{\gamma(y_t)}(y_t) \le \dist^{\gamma(z_t)}(z_t) - 2$ holds.
Thus we conclude $\dist^{\gamma(z_t)}(z_t) \le |P| = \dist^{\gamma(y_t)}(y_t) + 1 \le \dist^{\gamma(z_t)}(z_t) - 1$, a contradiction.

Otherwise, $\rho(t) = \OL{\gamma(y_t)}$.
By Lemma~\ref{lma:outgoingedgeproperty}(S4) again, we have
\begin{align*}
    \sumlevel(\out(t))
    &= \dist^{\rho(t)}(y_t) + \dist^{\rho(t)}(z_t)\\
    &\ge \dist^{\OL{\gamma(y_t)}}(y_t) + \dist^{\gamma(z_t)}(z_t)\\
    &> \dist^{\OL{\gamma(y_t)}}(y_t) + \dist^{\gamma(y_t)}(y_t),
\end{align*}
which implies that there exists an edge $e \in T^*(y_t)$ such that $\sumlevel(\out(t)) > \sumlevel(e)$ (since the union of shortest $\gamma(y_t)$- and $\OL{\gamma(y_t)}$-alternating paths from $f$ to $y_t$ must contain an edge in $T^*(y_t)$ of sum-level at most $\dist^{\OL{\gamma(y_t)}}(y_t) + \dist^{\gamma(y_t)}(y_t)$).
Hence $\level(\out(t)) > \level(\out(y_t))$, a contradiction.

Now we assume that $y_t$ is not an ancestor of $z_t$, which means that $\out(t)$ is also an outgoing edge of $T_{y_t}$.
Hence $\level(\out(t)) \geq \level(\out(y_t))$ necessarily holds and it suffices to consider the case of $\level(\out(t)) = \level(\out(y_t))$, i.e., $y_t \in \mathcal{V}_k$.
Since $\out(t)$ is the edge incident to $y_t$, $\out(y_t)$ must be incident to $y_t$ by the rule of choosing canonical MOEs.
By applying (S1) of this lemma (recall that the proof of (P2) assumes that (P1) holds even for $t \in \mathcal{V}_k$), we have a shortest $\rho(y_t)$-alternating 
path $Z'$ from $f$ to $y_t$ of level less than $\level(\out(y_t)) = \level(\out(t))$. By Lemma~\ref{lma:MOEparity}, 
$\level(\out(t)) = \level(\out(y_t))$ implies $\rho(t) = \rho(y_t)$. Hence $Z'$ is the one satisfying the condition of (S1).

\noindent
\textbf{(Case 2)} $\level(\out(t)) > \level(\out(y_t))$: In this case, $y_t \in \bigcup_{k' < k} \mathcal{V}_{k'}$ holds. 
By the induction hypothesis, we obtain a shortest $\rho(y_t)$-alternating path $Z'$ from $f$ to $z_{y_t}$ of level less than 
$\level(\out(y_t)) < \level(\out(t))$. Since $\level(Z') < \level(\out(y_t))$ holds, $Z'$ must contain $\{\parent(y_t), y_t\}$ to 
reach the inside of $T_{y_t}$, and contains no outgoing edge in $E^{\ast}(T_{y_t})$. Then $Z'[f, y_t]$ lies at the outside of $T_{y_t}$
and $Z'[y_t, z_{y_t}]$ lies at the inside of $T_{y_t}$. It implies that $Z'[f, y_t]$ is a shortest $\gamma(y_t)$-alternating path 
from $f$ to $y_t$ because if there exists a $\gamma(y_t)$-alternating path $Q$ from $f$ to $y_t$ shorter than $Z'[f, y_t]$, 
$Q \circ Z'[y_t, z_{y_t}]$ becomes a $\rho(y_t)$-alternating path shorter than $Z'$ (note that $Q$ lies at the outside of $T_{y_t}$
by Lemma~\ref{lma:outgoingedgeproperty}(S5) and thus $Q \circ Z'[y_t, z_{y_t}]$ necessarily becomes an alternating path).
In addition, by the induction hypothesis, 
we also has a shortest $\OL{\gamma(y_t)}$-alternating path $X'$ from $f$ to $y_t$ of level less than 
$\level(\out(y_t)) < \level(\out(t))$. Consequently, either $Z'[f, y_t]$ or $X'$ has parity $\rho(t)$, which is
the path satisfying the condition of (S1).

\noindent
\textbf{The proof of (P3)}: Suppose the existence of $Y'$ and $Z'$ stated in (S1) and (S2) for $t \in \mathcal{V}_k$. 
Since the levels of $Y'$ and $Z'$ are less than $\out(t)$, they respectively lie at the outside of $T_t$ and
at the inside of $T_t$. Hence the path $X' = Y' \circ \out(t) \circ \OL{Z'}[z_t, t]$ is an alternating path from $f$ to $t$. 
Since $Z'[f, t]$ is an alternating path of level less than $\level(\out(t))$, by Lemma~\ref{lma:outgoingedgeproperty}(S1) and (S4), 
$Z'[f, t]$ is a shortest $\gamma(t)$-alternating path from $f$ to $t$.
Hence $X'$ is a $\OL{\gamma(t)}$-alternating path, and we have $|\OL{Z'}[z_t, t]| \leq |Z'| - |Z[f, t]| = \dist^{\rho(t)}(z_t) - \dist^{\gamma(t)}(t)$. 
Also, $|Y'| = \dist^{\rho(t)}(y_t)$ holds. 
By Lemma~\ref{lma:Y2}(S2), $\dist^{\OL{\gamma(t)}}(t) \geq \dist^{\rho(t)}(z_t) + 
\dist^{\rho(t)}(y_t) + 1 - \dist^{\gamma(t)}(t)$ holds. 
Combining the three (in)equalities, we conclude 
$|X'| = |Y'| + 1 + |\OL{Z'}[z_t, t]| \leq \dist^{\OL{\gamma(t)}}(t)$.
\end{proof}

\subsection{Correctness of $\FUNC$}
\label{sec:extpath}

For the proof of Lemma~\ref{lma:augpathcorrectness}, we present two auxiliary lemmas.

\begin{lemma}
\label{lma:extsamelength}
For any $\theta$-extendable pair $(s, t)$, every $\theta$-extendable path from $s$ to $t$ has the same length.
\end{lemma}

\begin{proof}
Let $P$ and $Q$ be two different $\theta$-extendable paths from $s$ to $t$, and $P'$ and $Q'$ be their $\theta$-extensions. 
Since the levels of $P'$ and $Q'$ are both less than $\out(s)$, they contain $\{\parent(s), s\}$ and 
do not contain any outgoing edge of $T_s$. It implies that $P'[f, \parent(s)]$ and $Q'[f, \parent(s)]$ do not contain 
any vertex in $T_s$, and $P'[s, t]$ and $Q'[s, t]$ consist only of the vertices in $T_s$. Suppose for contradiction that 
$P = P'[s, t]$ and $Q = Q'[s, t]$ have different lengths. Without loss of generality, we suppose 
$|P'[s, t]| < |Q'[s, t]|$. Since $|P'| = |Q'|$ holds, we have $|P'[f, \parent(s)]| > |Q'[f, \parent(s)]|$. 
It implies that the path $Q'[f, \parent(s)] \circ \{\parent(s), s\} \circ Q'[s, t]$ 
becomes a $\theta$-alternating path shorter than $P'$ and $Q'$. It is a contradiction.
\end{proof}

For each $(s, t, \theta)$ such that $(s, t)$ is $\theta$-extendable, pick a minimum-level $\theta$-extendable path $X_{s, t, \theta}$ from $s$ to $t$ (with any tie-breaking rule).
Let $\mathcal{Q}_i$ be the set of all tuples $(s, t, \theta)$ such that $|X_{s, t, \theta}| = i$. Note that 
$\mathcal{Q}_i$ to which $(s, t, \theta)$ belongs is not affected by the tie-breaking rule for the choice of $X_{s, t, \theta}$ because of Lemma~\ref{lma:extsamelength}. 

\begin{lemma}
\label{lma:orthodoxpathparent}
For any $i$ and $(s, t, \gamma(t)) \in \mathcal{Q}_i$, $\{\parent(t), t\} \in E(X_{s, t, \gamma(t)})$ holds.
\end{lemma}

\begin{proof}
If $E^{\ast}(T_t)$ is empty, $\{\parent(t), t\}$ is a bridge of $G$ separating $s$ and $t$, and thus $X_{s, t, \gamma(t)}$ necessarily contains 
$\{\parent(t), t\}$.
Suppose that $E^{\ast}(T_t) \neq \emptyset$ and for contradiction that $\{\parent(t), t\} \not\in E(X_{s, t, \gamma(t)})$.
Then $X_{s, t, \gamma(t)}$ must contain an edge in $E^{\ast}(T_t)$ 
to reach vertices in $T_t$. Then we have $\level(X_{s, t, \gamma(t)}) \geq \out(t)$. Since 
$\level(X_{s, t, \gamma(t)}) \leq \level(W) < \level(\out(s))$ by definition, where $W$ is any $\gamma(t)$-extension of $X_{s, t, \gamma(t)}$, it follows $\level(\out(t)) < \level(\out(s))$. Lemma~\ref{lma:twopaths}(S2) 
states that there exists an alternating path 
$Z'$ from $f$ to $z_t$ of level less than $\level(\out(t))$. Since $\level(Z') < \level(\out(t)) < \level(\out(s))$ holds and $s$ and $t$ are the ancestors of $z_t$, both $s$ and $t$ are contained in $Z'$ by Lemma~\ref{lma:outgoingedgeproperty}(S1).
By lemma~\ref{lma:outgoingedgeproperty}(S3), any alternating path from $f$ to $t$ terminating with edge $\{\parent(t), t\}$ is 
a $\gamma(t)$-alternating path. It follows that the subpath of $Z'$ from $s$ to $t$ is a 
$\gamma(t)$-extendable path of level less than $\level(\out(t)) \leq \level(X_{s, t, \gamma(t)})$, which contradicts the definition that 
$X_{s, t, \gamma(t)}$ is the minimum-level.
\end{proof}

Utilizing the lemmas above, we prove Lemma~\ref{lma:augpathcorrectness}.


\let\temp\thelemma
\renewcommand{\thelemma}{\ref*{lma:augpathcorrectness}}
\begin{lemma}[Repeat]
For any tuple $(s, t, \theta)$ such that $(s, t)$ is $\theta$-extendable, $\FUNC(s, t, \theta)$ outputs a minimum-level $\theta$-extendable path from $s$ to $t$.
\end{lemma}
\let\thelemma\temp
\addtocounter{lemma}{-1}

\begin{proof}
We prove the lemma for any $i$ and $(s, t, \theta) \in \mathcal{Q}_i$ following the induction on $i$.
In the following argument, let $Y$ and $Z$ be the paths computed in $\FUNC(s, t, \theta)$ (see Algorithm~\ref{alg:func}), and $W$ be any $\theta$-extension of $X_{s, t, \theta}$. 

\noindent
\textbf{(Basis)} For $i = 0$, the algorithm returns the path of length zero, and thus the lemma obviously holds.

\noindent
\textbf{(Inductive step)} Suppose as the induction hypothesis that the lemma holds for any $i' < i$ and 
$(s', t', \theta') \in \mathcal{Q}_{i'}$, and consider a tuple 
$(s, t, \theta) \in \mathcal{Q}_i$. We divide the proof into two cases according to the value of $\theta$:

\noindent
\textbf{(Case 1)} $\theta = \gamma(t)$: 
First, we prove $(s, \parent(t), \OL{\theta}) \in \mathcal{Q}_{i-1}$.
By Lemma~\ref{lma:orthodoxpathparent}, $\parent(t)$ is on the path $X_{s, t, \theta}$, and thus on the path $W$. Then 
$W[f, \parent(t)]$ is a shortest $\OL{\theta}$-alternating path from $f$ to $\parent(t)$ of $\level(W) < \level(\out(s))$ 
because of the equality $|W[f, \parent(t)]| = |W| - 1 = \dist^{\theta}(t) - 1 = \dist^{\OL{\theta}}(\parent(t))$ derived 
from the definition of alternating base trees. It implies that $W[f, \parent(t)]$ is an $\OL{\theta}$-extension of 
$(s, \parent(t))$, i.e., $(s, \parent(t))$ is $\OL{\theta}$-extendable. In addition, $X_{s, t, \theta}[s, \parent(t)]$ is 
a $\OL{\theta}$-extendable path from $s$ to $\parent(t)$ of length $|X_{s, t, \theta}| - 1$,
i.e., $(s, \parent(t), \OL{\theta}) \in \mathcal{Q}_{i-1}$. 

By the induction hypothesis,  
$Y = \FUNC(s, \parent(t), \OL{\theta})$ is a minimum-level $\OL{\theta}$-extendable path from $s$ to $\parent(t)$. By Lemma~\ref{lma:outgoingedgeproperty}(S2), $Y$ does not contain $t$. 
Thus $P = Y \circ \{\parent(t), t\} = \FUNC(s, t, \theta)$ is an alternating path from $s$ to $t$. 
It is necessarily a minimum-level $\theta$-extendable path: We have already shown that 
$X_{s, t, \theta}[s, \parent(t)]$ is a $\OL{\theta}$-extendable path from $s$ to $\parent(t)$. 
Since $Y$ is the minimum-level, $\level(Y) \leq \level(X_{s, t, \theta}[s, \parent(t)])$ holds. 
The levels of $X_{s, t, \theta}$ and $P$ are respectively decided by the levels of their prefixes 
$X_{s, t, \theta}[s, \parent(t)]$ and $P[s, \parent(t)] = Y$ (recall that 
any edge in $T$ has level $(0, 0)$). Hence we can conclude $\level(X_{s,t,\theta}) = 
\level(X_{s, t, \theta}[s, \parent(t)]) \geq \level(Y) = \level(P)$. 
Since $\level(P) < \level(\out(s))$ holds, $P$ lies at the inside of $T_s$. Then 
$P' = W[f, s] \circ P$ becomes an alternating path of length $|W|$ because 
$|Y| = |W[s, \parent(t)]|$ holds by Lemma~\ref{lma:extsamelength}, i.e., $P$ admits a 
$\theta$-extension $P'$ and thus is a $\theta$-extendable path. 

\noindent
\textbf{(Case 2)} $\theta = \OL{\gamma(t)}$: By Lemma~\ref{lma:outgoingedgeproperty}(S3), $W$ does 
not contain $\{\parent(t), t\}$. Applying Lemma~\ref{lma:outgoingedgeproperty}(S1), we obtain 
$\level(\out(t)) \leq \level(W) < \level(\out(s))$. 
Let $X'$, $Y'$, and $Z'$ be the alternating paths defined in Lemma~\ref{lma:twopaths}. Two alternating paths 
$Y' \circ \{y_t, z_t\}$ and $Z'$ from $f$ to $z_t$ have levels at most $\level(\out(t)) < \level(\out(s))$. 
Since $s$ is an ancestor of $t$ and thus that of $z_t$, 
$s \in V(Y')$ holds by Lemma~\ref{lma:outgoingedgeproperty}(S1).
Then $Y'$ and $Z'$ are respectively 
$\rho(t)$-extensions of $(s, y_t)$ and $(t, z_t)$, and $Y'[s, y_t]$ and $Z'[t, z_t]$ are 
the corresponding $\rho(t)$-extendable paths. Since $X'$ is also a $\theta$-extension of $(s, t)$, 
$X'[s, t]$ is a $\theta$-extendable path from $s$ to $t$. 
It follows $|Y'[s, y_t]| + |Z'[t, z_t]| = |X_{s, y_t, \rho(t)}| + |X_{t, z_t, \rho(t)}|  < 
|X'[s, t]| = |X_{s, t, \theta}|$, where the first and last equalities come from 
Lemma~\ref{lma:extsamelength}. It implies $|Y'[s, y_t]| < |X_{s, t, \theta}|$ 
and $|Z'[t, z_t]| < |X_{s, t, \theta}|$, i.e., $(s, y_t), (t, z_t) \in \bigcup_{i' < i} 
\mathcal{Q}_{i'}$. By the induction hypothesis,  
$Y = \FUNC(s, y_t, \rho(t))$ and $Z = \FUNC(t, z_t, \rho(t))$ are respectively 
the minimum-level $\rho(t)$-alternating paths from $s$ to $y_t$ and from $t$ to $z_t$. 
Since $Y$ and $Z$ are the minimum-level, their levels are
respectively upper bounded by $\level(Y'[s, y_t])$ and $\level(Z'[t, z_t])$, and here, less than $\level(\out(t))$. Since we have already shown
$\level(\out(t)) < \level(\out(s))$, it also follows 
$\level(Y) < \level(\out(s))$ and $\level(Z) < \level(\out(s))$. 
Then $Y$ and $Z$ do not contain any edge in $E^{\ast}(T_s) \cup E^{\ast}(T_t)$, i.e., 
$Y$ lies at the inside of $T_s$ but the outside of $T_t$, and $Z$ lies at the inside of $T_t$. 
Consequently, $P = Y \circ \{y_t, z_t\} \circ \OL{Z}$ returned by the algorithm is an 
alternating path of level $\level(\out(t))$ lying at the inside of $T_s$. By 
Lemma~\ref{lma:extsamelength}, we have 
\begin{align*}
\dist^{\theta}(t) &= |X'| \\
& = |X'[f, s]| + |X'[s, y_t]| + 1 + |X'[z_t, t]|  \\
& = |X'[f, s]| + |Y| + |Z| + 1 \\
&= |X'[f, s]| + |P|.
\end{align*}
That is, $X'[f, s] \circ P$ is a shortest $\theta$-alternating path from $f$ to $t$.
The level of $X'[f, s] \circ P$ is bounded as follows:
\begin{align*}
\level(X'[f, s] \circ P) &= \max\{\level(X'[f, s]), \level(P)\} \\
&\leq 
\max\{\level(X'), \level(Y), \level(Z), \level(\{y_t, z_t\})\} \\
&= \level(\out(t)) \\
&< \level(\out(s)), 
\end{align*}
where we already see $\level(X') = \level(\{y_t, z_t\}) = \level(\out(t))$, $\level(Y) \leq \level(\out(t))$, and $\level(Z) < \level(\out(t))$.
It implies $X'[f, s] \circ P$ is a $\theta$-extension of 
$(s, t)$, and thus $P$ is a $\theta$-extendable path of level at most $\level(\out(t))$.
Since $W$ does not contain $\{\parent(t), t\}$ and $s$ is at the outside of $T_t$, 
$W[s, t] = X_{s, t, \theta}$ must contain at least one outgoing edge of $T_t$. Then 
$\level(X_{s, t, \theta}) \geq \level(\out(t))$ holds and thus the constructed path $P$ is 
of the minimum-level.
\end{proof}

\subsection{Running-Time Analysis of Distributed Implementation}



\let\temp\thelemma
\renewcommand{\thelemma}{\ref*{lma:token_time}}
\begin{lemma}[Repeat]
For any $\overline{\gamma(t)}$-extendable pair $(s,t)$, let $Z$ be a $\rho(t)$-extendable path of $(t,z_t)$ and $Y$ be a $\rho(t)$-extendable path of $(s,y_t)$. Then 
$|R_{t, z_t}| \leq |Z|$ and $|R_{t, z_t}| \leq |Y|$ hold.
\end{lemma}
\let\thelemma\temp
\addtocounter{lemma}{-1}

\begin{proof}
    By Lemma~\ref{lma:outgoingedgeproperty}(S4), $\dist^{\gamma(t)}(t) \geq \dist^{\gamma(\parent(t))}(\parent(t)) + 1$ holds
    for any $t \in V(G) \setminus \{f\}$. It implies
    $|R_{t, z_t}| \leq \dist^{\gamma(z_t)}(z_t) - \dist^{\gamma(t)}(t)$. Let $Z'$ be a  
    shortest $\rho(t)$-alternating path from $f$ to $z_t$ of level less than $\level(\out(t))$ by Lemma~\ref{lma:twopaths}. Then $|Z'| \geq 
    \dist^{\gamma(z_t)}(z_t)$ holds. Since $Z'[t, z_t]$ is a $\rho(t)$-extendable
    path, its length is equal to $|Z|$ by Lemma~\ref{lma:extsamelength}. Since $\level(Z') < \level(\out(t))$ holds,
    $Z'[f, t]$ must terminate with $\{\parent(t), t\}$ by Lemma~\ref{lma:outgoingedgeproperty}(S1), and thus 
    it is a $\gamma(t)$-alternating path by Lemma~\ref{lma:outgoingedgeproperty}(S3). Then 
    Lemma~\ref{lma:outgoingedgeproperty}(S5) implies $|Z'[f, t]| = \dist^{\gamma(t)}(t)$ because if it does not hold one can obtain a $\rho(t)$-alternating path from $f$ to $z_t$ shorter than $Z'$ by replacing the prefix $Z'[f, t]$ of $Z'$ with 
    the shortest $\gamma(t)$-alternating path from $f$ to $t$. Putting all the inequalities together,
    we obtain $|R_{t, z_t}| \leq \dist^{\gamma(z_t)}(z_t) - \dist^{\gamma(t)}(t) \leq |Z'| - |Z'[f, t]| = |Z'[t, z_t]| = |Z|$.

    Next, we consider the second inequality. Let $Y'$ be a shortest $\rho(t)$-alternating path from $f$ to $y_t$ of level 
    less than $\level(\out(t))$ by Lemma~\ref{lma:twopaths}. Since it lies at the outside of $T_t$, $Y' \circ \{y_t, z_t\}$ is
    an alternating path from $f$ to $z_t$. Then $\dist^{\gamma(z_t)}(z_{t}) \leq \dist^{\rho(t)}(y_t) + 1$ holds.
    Since $Y'[s, y_t]$ is a $\rho(t)$-extendable path, its length is equal to $|Y|$. Similarly, we obtain  
    $|Y'[f, s]| = \dist^{\gamma(s)}(s)$. Since $s$ is an ancestor of $t$,
    we obtain $\dist^{\gamma(s)}(s) < \dist^{\gamma(t)}(t)$ by Lemma~\ref{lma:outgoingedgeproperty}(S4). Putting all 
    the inequalities together, we have $|Y| = |Y'[s, y_t]| = |Y'| - |Y'[f, s]| = 
    \dist^{\rho(t)}(y_t) - \dist^{\gamma(s)}(s) > \dist^{\gamma(z_t)}(z_t) - 1 - \dist^{\gamma(t)}(t) \geq |R_{t, z_t}| - 1$.
    That is, $|Y| \geq |R_{t, z_t}|$ holds.
\end{proof}

The total running time is bounded as follows:


\let\temp\thelemma
\renewcommand{\thelemma}{\ref*{lma:exe_time}}
\begin{lemma}[Repeat]
Let $g(h)$ be the worst-case round complexity of $\FUNC(s,t,\theta)$ over all 
$(s, t, \theta)$ admitting a $\theta$-extendable path from $s$ to $t$ of length $h$.
For any $0 \leq h \leq \ell$, $g(h) \leq 2h$ holds.
\end{lemma}
\let\thelemma\temp
\addtocounter{lemma}{-1}

\begin{proof}
    We first see that any message transmitted in the algorithm does not suffer edge congestion,
    i.e., any two recursive invocations simultaneously executed do not simultaneously 
    send two different messages through the same link with the same direction. A simple inductive 
    argument shows that only vertices in $V(T_s)$ sends a message in the run of 
    $\FUNC(s, t, \theta)$. Given two recursive invocations $\FUNC(s, t, \theta)$ and 
    $\FUNC(s', t, \theta)$, $V(T_s)$ and $V(T_{s'})$ intersect if and only if $s$ and $s'$ have
    the ancestor-descendant relationship, which can occur only in the case of $\FUNC(s', t, \theta)$
    is called at the inside of $\FUNC(s, t, \theta)$ or vice versa. That is, they never simultaneously run, and thus any two invocations simultaneously running do not send two messages from the same vertex. 
    
    We prove the time bound claimed in the lemma following the observation above. Specifically, we bound
    the running time of $\FUNC(s, t, \theta)$ by $2h$ for any $(s, t, \theta)$ admitting a $\theta$-extendable path from $s$ to $t$ of length $h$.
    The proof is based on the induction on $h$.

    \noindent
    \textbf{(Basis)} The case of $h = 0$, the lemma obviously holds.

    \noindent 
    \textbf{(Inductive Step)} Suppose as the induction hypothesis that $g(h')\leq 2h'$ holds for any 
    $0 \leq h' < h$. We consider the following two cases:
    
    \noindent
    \textbf{(Case 1)} For $(s, t, \theta)$ such that $\theta=\gamma(t)$ holds: 
    By the correctness of $\FUNC$, $(s, \parent(t))$ admits a $\OL{\theta}$-extendable path of length less than $h$. The induction hypothesis bounds
    the running time of the recursive call $\FUNC(s, \parent(t), \OL{\theta})$ by $2(h - 1)$ rounds.
    Adding the cost of transmitting the triggering message (one round), we upper bound the running time
    of $\FUNC(s, t, \theta)$ by $2h$.
    
    \noindent
    \textbf{(Case 2)} For $(s, t, \theta)$ such that $\theta = \OL{\gamma(v)}$ holds:
    By the correctness of $\FUNC$, $(s, y_t)$ and $(t, z_t)$ admit $\rho(t)$-extendable paths $Y$
    and $Z$ respectively, which satisfy $|Y| + |Z| + 1 = h$. Then, by the induction hypothesis,
    $\FUNC(s,y_t,\rho(t))$ and $\FUNC(t,z_t, \rho(t))$ respectively run within 
    $2|Y|$ and $2|Z|$ rounds. Since they run in parallel, the running time for the recursive calls 
    is bounded by $2 \max\{|Y|, |Z|\}$. By Lemma~\ref{lma:token_time}, the length of 
    $R_{t,z_t} \circ \{y_t, z_t\}$, which is the cost for transmitting the triggering message, is at most 
    $\min\{|Y|+1, |Z|+1\}$. The total running time is at most $2 \max\{|Y|, |Z|\} + 
    \min\{|Y|+1, |Z|+1\} \leq 2(|Y| + |Z| + 1) = 2h$. 
\end{proof}

\section{Concluding Remarks}

In this paper, we presented a nearly linear-time exact algorithm for the maximum matching problem in
the CONGEST model. This algorithm substantially improves the current best upper bound. To 
construct our algorithm, we proved a new structural lemma of alternating base trees (Lemma~\ref{lma:twopaths}), which yields a very simple recursive construction of shortest augmenting 
paths in a fully-decentralized manner. We believe that this is a promising tool toward further 
improved algorithms. 

Finally, we close this paper with a short remark on designing further improved 
(i.e., sublinear-time) algorithms. A major open problem following our result is 
to develop an algorithm of finding an augmenting path of length $\ell$ within 
$O(\ell^{1 - \delta})$ rounds for a small constant $\delta > 0$. Since the iterative 
improvement of matchings can occur $\Theta(n)$ times, a straightforward application 
of such an algorithm does not yield any sublinear-time algorithm. 
Fortunately, this issue is resolved by the $(1 - \epsilon)$-approximate matching algorithm 
by Fischer, Slobodan, and Uitto running in $\mathrm{poly}(1/\epsilon, \log n)$ rounds: We first
run it with $\epsilon = n^{-\alpha}$ for a sufficiently small $\alpha > 0$ such that 
its running time becomes sublinear. Starting from the output of the algorithm, we apply 
the sequential iterative improvement by augmenting paths. Due to the guarantee of 
$(1 - \epsilon)$-approximation, the number of iterations is bounded by $O(n^{1 - \alpha})$. 
If one can utilize an 
$O(\ell^{1 - \delta})$-round augmenting path algorithm, the running time of whole iterative 
improvement process is sublinearly bounded. In summary, any $O(\ell^{1 - \delta})$-round 
algorithm of finding an augmenting path of length $\ell$ derives a sublinear-time CONGEST 
algorithm for the exact maximum matching problem.

\section*{Acknowledgments}
The first author was supported by JSPS KAKENHI Grant Numbers 23H04385, 22H03569, and 21H05854.
The second author was supported by JSPS KAKENHI Grant Numbers 22K21277 and 23K16838.
The third author was supported by JSPS KAKENHI Grant Numbers 20K19743 and 20H00605.

\nocite{*}


\appendix

\section{Proof Outline of Lemmas~\ref{lma:HK-framework} and \ref{lma:AK-framework}}
\label{appendix:frameworks}

\subsection{Lemma~\ref{lma:HK-framework}}
\label{sec:highlevel}

Let $\SAPath(M, \ell)$ be a CONGEST algorithm with the following property:
\begin{quote}
For any graph $G = (V, E)$ and matching $M \subseteq E$,
$\SAPath(M, \ell)$ finds a nonempty set of vertex-disjoint augmenting paths with high probability within $O(\ell)$ rounds if $(G, M)$ has 
an augmenting path of length at most $\ell$. Each vertex $u \in V(G)$ outputs the predecessor and successor of the output augmenting path 
to which $u$ belongs (if it exists).
\end{quote}
The construction of the algorithm $\SAPath$ is discussed in Section~\ref{sec:augpathconstruction}. 
The pseudocode of whole algorithm is presented in Algorithm~\ref{alg:mm}. It basically follows the standard idea of centralized maximum 
matching algorithms, i.e., finding an augmenting path and improving the current matching iteratively. In the $i$-th iteration, the algorithm $\SAPath(M, \ell)$ runs with $\ell=\lfloor2\AMmax/(\AMmax-i)\rfloor$. This setting comes from Theorem~\ref{thm:hk}. 
The improvement of the current matching by a given augmenting path is simply a local operation and is realized by flipping the labels of matching edges and non-matching edges on the path. The correctness and running time of Algorithm~\ref{alg:mm} are stated as follows:

\begin{lemma}
\label{lma:HK}
Assume that there exists an algorithm $\SAPath(M, \ell)$ specified above. Then
Algorithm~\ref{alg:mm} constructs a maximum matching with high probability in $\tilde{O}(\Mmax)$ rounds.
\end{lemma}
\begin{proof}
The running time of the algorithm is bounded as follows (recall $\AMmax = O(\Mmax)$):
\begin{align*}
& O\left(\sum_{i=0}^{\AMmax - 1}\left(\left\lfloor \frac{2\AMmax}{\AMmax-i} \right\rfloor\right)\right) = 
O\left(\AMmax \sum_{i=1}^{\AMmax}\frac{1}{i} \right) = O(\AMmax \log \AMmax) = \tilde{O}(\Mmax).
\end{align*}

We consider the correctness of the algorithm.
Let $s(i)$ be the matching size at the beginning of the $i$-th iteration.
To prove the correctness, it suffices to show that $s(\AMmax - \Mmax + i + 1) > s(\AMmax - \Mmax + i)$ for any $i$ satisfying $s(\AMmax - \Mmax + i) = i < \Mmax$.
As $\AMmax \ge \Mmax$, this inductively implies that $s(\AMmax - \Mmax + i) \ge i$ for any $0 \le i \le \Mmax$, and hence $s(\AMmax) = \Mmax$.

Suppose that $s(\AMmax - \Mmax + i) = i < \Mmax$.
Then, in the $(\AMmax - \Mmax + i)$-th iteration, there exists an augmenting path of length at most 
$\lfloor 2\Mmax/(\Mmax - i) \rfloor \leq \lfloor 2\AMmax/(\Mmax - i) \rfloor = \lfloor 2\AMmax/(\AMmax - (\AMmax - \Mmax + i)) \rfloor$ by Theorem \ref{thm:hk}.
Thus the algorithm necessarily finds at least one augmenting path in this iteration, which implies $s(\AMmax - \Mmax + i + 1) > s(\AMmax - \Mmax + i)$.
\end{proof}

\begin{algorithm*}[t]
\caption{Constructing a maximum matching in $\tilde{O}(\Mmax)$ rounds.}
\label{alg:mm}
\begin{algorithmic}[1]
\FOR{$i=0;i < \AMmax;i++$}
\STATE run the algorithm $\SAPath(M,\ell)$ with $\ell = \lfloor2\AMmax/(\AMmax-i)\rfloor$ for $O(\ell)$ rounds.
\IF{$\SAPath(M,\ell)$ finds a nonempty set of vertex-disjoint augmenting paths within $O(\ell)$ rounds,}
\STATE improve the current matching using the set of vertex-disjoint augmenting paths.
\ENDIF
\ENDFOR
\end{algorithmic}
\end{algorithm*}

\subsection{Lemma~\ref{lma:AK-framework}}
\label{sec:verification}
Although the original algorithm is designed for the verification of maximum matching, it provides 
each vertex with information on the length of alternating paths to the closest unmatched vertex.  
The following theorem holds.

\begin{theorem}[Ahmadi and Kuhn~\cite{AK20}]
\label{theo:verification}
Assume that a graph $G = (V, E)$ and a matching $M \subseteq E$ are given, and let 
$W$ be the set of all unmatched vertices. There exist two $O(\ell)$-round randomized CONGEST algorithms $\MV(M, \ell, f)$ and $\PART(M, \ell)$ that output the following information at every vertex $v \in V(G)$
with a probability of at least $1-1/n^c$ for an arbitrarily large constant $c > 1$.
\begin{enumerate}
\item Given $M$, a nonnegative integer $\ell$, and a vertex $f \in W$, 
$\MV(M, \ell, f)$ outputs the pair $(\theta, \dist^{\theta}_G(f, v))$ at each vertex $v \in V(G)$
if $\dist^{\theta}_G(f, v)\leq \ell$ holds (if the condition is satisfied for both $\theta = \Odd$ and $\theta = \Even$, $v$ outputs two pairs). The algorithm $\MV(M, \ell, f)$ is initiated only by the vertex $f$ (with the value $\ell$), and other vertices do not require information on the ID of $f$ and value $\ell$ at the beginning of the algorithm.
\item The algorithm $\PART(M, \ell)$ outputs a family of disjoint subsets $V^1, V^2, \dots, V^{N}$ of $V(G)$ (as the label $i$ for each vertex in $V^i$) such that (a) the subgraph $G^i$ induced by
$V^i$ contains exactly two unmatched vertices $f^i$ and $g^i$ and an augmenting path between $f^i$ and $g^i$ of length at most $\ell$, (b) for every vertex $v \in V^i$, the length of the shortest alternating path from $f^i$ to $v$ is at most $\ell$ (and hence the diameter of $G^i$ is $O(\ell)$), and (c) every vertex in $V^i$ knows the ID of $f^i$.
\end{enumerate}
\end{theorem}

While the original paper~\cite{AK20} only presents a single algorithm returning the outputs of both $\MV$ and $\PART$, in a slightly weaker form, it is very straightforward to extend the result to the algorithms above (see \cite{KI22} for the details of this extension). 
It is easy to transform any algorithm $\SAPath'$ of Lemma~\ref{lma:AK-framework} to the algorithm $\SAPath$ satisfying the specification
of Lemma~\ref{lma:HK-framework}. The transformation just follows (1) running $\PART(M, \ell)$ and 
(2) executing the algorithms $\MV(M, \ell, f)$ and $\SAPath'$ at each subgraph $G^i$ independently. 

\section{Precomputing Phase of Distributed Implementation}
\label{appendix:aggregate}

We recall the precomputed information required by $\FUNC$:

\begin{itemize}
    \item The parent and children of $u$ in $T$, and the value of $\gamma(u)$.
    \item The IDs of two endpoints of $\out(u)$ (if $\out(u)$ is not virtual).  
    \item Let $R_{v, z_v}$ be the shortest path from $v$ to $z_v$ in $T$. If $u \in V(R_{v, z_v})$, $u$ has the information for routing messages from 
    $v$ to $z_v$. More precisely, $u$ knows the pair $(z_v, u')$, where $u'$ is the successor of $u$ in 
    $R_{v, z_v}$.
\end{itemize}

The alternating base tree $T$ is constructed within $O(1)$ rounds by exchanging the information of 
$\dist^{\Odd}(\cdot)$ and 
$\dist^{\Even}(\cdot)$ between every neighboring vertex pair, because each vertex can determine 
its parent in $T$ from the exchanged information. We refer to the height of $T$ as $\ell'$.
Obviously $\ell' = O(\ell)$ holds.

The information of $\out(u)$ is obtained by the pipelined 
broadcast and aggregation over $T$: First, each node $u$ constructs the list of all its ancestors. 
This task is implemented by the pipelined (downward) broadcast in $T$, which takes $O(\ell)$ rounds.
Then the constructed list is exchanged through every non-tree edge $e = \{x, y\}$. More precisely, $x$ and $y$ 
exchange the information of the ancestor list, $\dist^{\Odd}(\cdot)$, and $\dist^{\Even}(\cdot)$. 
Since the list size is bounded by the height of $T$, i.e., $O(\ell)$, this information exchange also finishes within 
$O(\ell)$ rounds. Then $x$ and $y$ can compute the level of $e$, and the ID of the least common ancestor of 
$x$ and $y$ (denoted by $\lca(e)$). We define $F_i$ as the set of non-tree edges $e$ such that the depth of 
$\lca(e)$ in $T$ is $i$. Note that each node $u$ can decide if $e \in F_i \cap E^{\ast}(T_u)$ holds or not only from the information of $\lca(e)$ and its depth, because $e \in F_i \cap E^{\ast}(T_u)$ holds 
if and only if the depth of $\lca(e)$ is $i$ and $\lca(e)$ is an ancestor of $u$. 
For each $0 \leq i \leq \ell' - 1$, the algorithm informs each vertex $u$ 
of the minimum-level edge in $F_i \cap E^{\ast}(T_u)$ by a single-shot tree aggregation over $T$.
The behavior of each vertex $u$ in the $i$-th aggregation 
process is explained as follows:
\begin{itemize}
\item The vertex $u$ is a leaf: $u$ chooses the minimum-level edge in  
$I(u) \cap (F_i \cap E^{\ast}(T_u)) = F_i \cap  I(u)$. If $F_i \cap I(u)$ is empty, 
the dummy edge of level 
$(\infty, \infty)$ is chosen. If two or more minimum-level edges exist, an arbitrary one of
them is chosen.
\item The vertex $u$ is internal: Let $X$ be the set of edges received from its children. The vertex
$u$ chooses the minimum-level edge in $(X \cup I(u)) \cap (F_i \cap E^{\ast}(T_u))$. If two or more 
minimum-level edges exist, the edge incident to $u$ has a higher priority. For two edges with the same priority
(i.e., two edges incident to $u$, or two edges not incident to $u$), we use an arbitrary tie-breaking rule.
Let $\{v, w\}$ be the chosen edge, where $w$ is the vertex in $T_u$. If $w \neq u$ holds, the information of $\{v, w\}$ 
comes from a child vertex $u'$ of $u$, and thus $u$ stores the routing information $(w, u')$.
\end{itemize}

By the standard pipeline technique, one can finish the tasks for all $i$ in $O(\ell)$ rounds. 
Finally, each vertex $u$ computes its minimum outgoing edge as follows:
\begin{align*}
\out(u) &= \mathrm{minedge} \left( \bigcup_{i} F_i \cap E^{\ast}(T_u)\right) \\
        &= \mathrm{minedge} \left( \bigcup_{i} \mathrm{minedge}\left(F_i \cap E^{\ast}(T_u)\right)\right),
\end{align*}
where $\mathrm{minedge}(F')$ is the minimum-level edge in $F'$, and the tie-breaking rule respects the rule of choosing canonical MOEs in Section~\ref{sec:abtreeandlevel}.

%

%
%

%

\end{document}